%% file: aaai23_main.tex
\documentclass[letterpaper]{article} %
\usepackage{aaai23}  %
\usepackage{times}  %
\usepackage{helvet}  %
\usepackage{courier}  %
\usepackage[hyphens]{url}  %
\usepackage{graphicx} %
\usepackage{natbib}  %
\usepackage{caption} %
\usepackage[utf8]{inputenc} %
\usepackage{url}            %
\usepackage{booktabs}       %
\usepackage{amsfonts}       %
\usepackage{nicefrac}       %
\usepackage{microtype}      
\usepackage{wrapfig}

\usepackage{enumitem}

\usepackage{amsmath}
\usepackage{amsfonts} 
\usepackage{amsthm}	%
\usepackage{dsfont}
\usepackage{bm} %
\newtheorem{definition}{Definition}[]
\newtheorem{theorem}{Theorem}
\newtheorem{lemma}[theorem]{Lemma}

\newtheorem{corollary}{Corollary}

\newtheorem{condition}{Condition}
\newtheorem{observation}{Observation}

\usepackage{thm-restate}
\usepackage{theoremref} %
\usepackage{amsxtra} %
\usepackage{graphicx}
\usepackage{caption, subcaption}

\usepackage{algorithm}
\usepackage[ruled,vlined,algo2e,linesnumbered]{algorithm2e}

\DeclareMathOperator*{\argmax}{\arg\!\max}

\DeclareMathOperator*{\E}{\mathbb{E}}

\title{DM$^2$: Decentralized Multi-Agent Reinforcement Learning via Distribution Matching}

\author {
    Caroline Wang \textsuperscript{\rm 1}\thanks{Equal contribution.},
    Ishan Durugkar \textsuperscript{\rm 1}\footnotemark[1],
    Elad Liebman \textsuperscript{\rm 2}\footnotemark[1],
    Peter Stone \textsuperscript{\rm 1,\rm 3}
}
\affiliations {
    \textsuperscript{\rm 1} The University of Texas at Austin\\
    \textsuperscript{\rm 2} SparkCognition Research\\
    \textsuperscript{\rm 3} Sony AI\\
    caroline.l.wang@utexas.edu, ishand@cs.utexas.edu, \\ eliebman@sparkcognition.com, pstone@cs.utexas.edu
}

\begin{document}

\maketitle

\input{content_abstract}
\input{content_main}

\bibliography{refs}

\onecolumn
\appendix
\input{content_appendix}

\end{document}

%% file: content_abstract.tex
\begin{abstract}
Current approaches to multi-agent cooperation rely heavily on centralized mechanisms or explicit communication protocols to ensure convergence. 
This paper studies the problem of distributed multi-agent learning without resorting to centralized components or explicit communication.
It examines the use of distribution matching to facilitate the coordination of independent agents.
In the proposed scheme, each agent independently minimizes the distribution mismatch to the corresponding component of a target visitation distribution.
The theoretical analysis shows that under certain conditions, each agent minimizing its individual distribution mismatch allows the convergence to the joint policy that generated the target distribution.
Further, if the target distribution is from a joint policy that optimizes a cooperative task, the optimal policy for a combination of this task reward and the distribution matching reward is the same joint policy.
This insight is used to formulate a practical algorithm (\textsc{dm$^2$}), in which each individual agent matches a target distribution derived from concurrently sampled trajectories from a joint expert policy.
Experimental validation on the StarCraft domain shows that combining (1) a task reward, and (2) a distribution matching reward for expert demonstrations for the same task, allows agents to outperform a naive distributed baseline.
Additional experiments probe the conditions under which expert demonstrations need to be sampled to obtain the learning benefits.
\end{abstract}

%% file: content_main.tex
\section{Introduction}
\label{introduction}
Multi-agent reinforcement learning (MARL) \citep{littman1994markov} is a paradigm for learning agent policies that may interact with each other in cooperative or competitive settings \citep{silver2017go2, silver2018alphazero,barrett2012analysis, leibo2017multi}.
Training multiple agents at once is challenging, since an agent updating its own strategy induces a nonstationary environment for other agents, potentially leading to training instabilities, and offsetting any theoretical guarantees single agent RL algorithms confer.
To overcome these issues, agent policies can be set up as a single, centralized joint policy, be trained together but then deployed individually \citep{rashid18qmix, foerster18coma}, or be coordinated through some form of communication \citep{Lowe2017MultiAgentAF, jaques19socialinfluence, liu2021coachplayer}.

Fully distributed training of agent policies remains an open problem in MARL.
Distributed, or decentralized, training is desirable particularly in 
situations where parallelism, robustness, flexibility, or scalability is needed. Such settings include where there are a large number of agents, where agents are faced with changing environments \citep{marinescu2017prediction}, where agents must perform tasks in varying team configurations over their lifetime \citep{thrun1998lifelong}, or where ensuring privacy is a concern \citep{leaute2013protecting}.

This paper considers the setting of cooperative tasks involving $K$ agents, where the goal is to learn a high-performing joint policy in a fully distributed fashion.
To mitigate the limitations imposed by this setting, we propose \emph{distribution matching} to a target state-action distribution, as a strategy to induce coordination.
We assume that this distribution is associated with the execution of some joint policy, such as demonstrations of expert teams \citep{song18magail}, or from high-performing trajectories from the agents' past interactions \citep{hao2019independent}.
One tempting way to utilize this distribution is to assign each agent the distribution associated with its corresponding expert, and have the agent minimize the distribution mismatch to this target distribution over states and actions.

At first glance, this approach is fraught with complications.
Since the target distribution over states and actions is based on the execution of some joint policy, a single agent trying to adjust its policy
might not make meaningful progress on its own, given that other agents could change their behaviors at the same time. 
Second, this distributed approach to distribution matching could suffer from the same destabilization that causes distributed MARL to diverge \citep{HernandezLeal2019ASA, Yang2020AnOO}.

This paper shows that despite the above complications, individual distribution matching can be combined with maximization of shared task rewards to learn effectively.
In particular, the contributions of the paper are:
\begin{itemize}[noitemsep,topsep=0pt,leftmargin=*]
    \item Theoretical analysis showing that distributed distribution matching to the target distributions converges to the joint expert policy that generated the demonstrations.
    \item The \textsc{dm$^2$} algorithm, a practical method combining the distribution matching reward with the task reward. Experimental validation shows that if demonstrations are aligned with the shared objective, \textsc{dm$^2$} accelerates learning compared to a decentralized baseline learning with the task reward only.
    \item Ablations that empirically verify our assumption that the target distribution needs to be induced by demonstrations from coordinated policies, but do not necessarily need to be concurrently sampled.
\end{itemize}

\begin{figure}[tb] %
    \centering
    \includegraphics[width=0.8\linewidth]{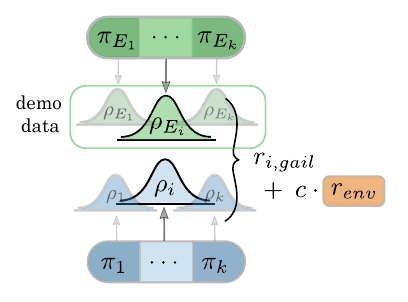}
    \caption{In \textsc{dm$^2$}, each agent $\pi_i$ independently learns from the sum of a distribution matching reward, $r_{i, gail}$, and a shared task reward, $r_{env}$. The distribution matching reward is computed by comparing the marginal state visitation distribution of the agent, $\rho_i$, with the state visitation distributions implied by the corresponding expert $\pi_{E_i}$'s demonstrations . 
    }
    \label{fig:method}
\end{figure}

\section{Related Work}
\label{section:related_work}
This section details related work, separated into work that relates to decentralized learning and that relates to distribution matching.

\paragraph{Cooperation in the Decentralized Setting:}
Many algorithms for multi-agent cooperation tasks require some degree of information sharing between agents.
 Centralized training decentralized execution (CTDE) methods use a single centralized critic that aggregates information during training, but is no longer required at execution time \citep{Lowe2017MultiAgentAF, sunehag18vdn, rashid18qmix, foerster18coma, yu2021surprising}.
 In practical implementations, agent networks often share parameters during training as well.

Rather than sharing model components, methods may also explicitly communicate  information between agents.
 Agents may be allowed to directly communicate information to each other  \citep{jaques19socialinfluence, li20matrpo, konan22midecmarl}.
 There might also be a central network that provides coordinating signals to all agents \citep{He2020LearningTC, liu2021coachplayer}.
 Knowledge of other agents' policies during training may also be assumed to limit the deviation of the joint policy \citep{wen2021multiagent}.
 
This work studies the fully decentralized setting without communication or shared model components. To our knowledge, relatively few works consider this setting.
Early work analyzed simple cases where two agents with similar but distinct goals could cooperate for mutual benefit under a rationality assumption \citep{Rosenschein1989CommunicationFreeIA, Genesereth1986CooperationWC}.
More recently, in the ALAN system for multi-agent navigation \citep{Godoy2018ALANAL}, agents learn via a multi-armed bandits method that does not require any communication.
\citet{Jiang2021OfflineDM} study the decentralized multi-agent cooperation in the \textit{offline} setting---in which each agent can only learn from its own data set of pre-collected behavior without communication---and propose a learning technique that relies on value and transition function error correction.

\paragraph{Distribution Matching in MARL:}
\citet{ho2016generative} originally proposed adversarial distribution matching  as a way to perform imitation learning in the single agent setting (the \textsc{gail} algorithm).
\citet{song18magail} extend \textsc{gail} to the multi-agent setting in certain respects. Their analysis sets up independent imitation learning as searching for a Nash equilibrium, and assumes that a unique equilibrium exists. Their experiments focus on training the agent policies in the CTDE paradigm, rather than the fully distributed setting.
This work instead leverages recent single-agent \textsc{gail} convergence theory \citep{Guan2021WhenWG} to demonstrate convergence to the joint expert policy, and performs experiments with distributed learning.
\citet{Wang21MAILCopulas} study MARL using copula functions to explicitly model the dependence between marginal agent policies for multi-agent imitation learning.
\citet{durugkar2020balancing} and \citet{Radke2022ExploringTB} show that  balancing individual preferences (such as matching the state-action visitation distribution of some strategies) with the shared task reward can accelerate progress on the shared task. 
In contrast to these works, the goal of this paper is not to study imitation learning, but rather to study how distribution matching by independent agents can enhance performance in cooperative tasks.

Perhaps most closely related to this work, \citet{hao2019independent} use self-imitation learning (\textsc{sil}) \citep{oh2018SIL} to encourage agents to repeat actions that led to high returns in the past.
The above approach can be considered as a special case of the setting this paper studies, where the target distribution can be non-stationary, and is generated by the agents themselves.
This paper further presents a theoretical analysis, showing that in the case where the target distribution is generated by demonstrations (and is therefore stationary), each agent attempting to minimize mismatch to their individual target distributions leads to convergence to the joint target policy.
Due to the non-stationary nature of the target distribution in \textsc{sil}, similar guarantees cannot be obtained.

\section{Background}
\label{section:background}
This section describes the problem setup for MARL, imitation learning, and distribution matching.

\paragraph{Markov games:}
A Markov game \citep{littman1994markov} or a stochastic game \citep{Gardner1983GameTheory} with $K$ agents is defined as a tuple $\langle K, \mathcal{S}, \bm{\mathcal{A}}, \rho_0, \mathcal{T}, \bm{R}, \gamma \rangle$, where $\mathcal{S}$ is the set of states, and $\bm{\mathcal{A}}\equiv \mathcal{A}^K$ is the product of the set of actions $\mathcal{A}$ available to each agent.
The initial state distribution is described by $\rho_0: \mathcal{S} \mapsto \Delta(\mathcal{S})$, where $\Delta(\cdot)$ indicates a distribution over the corresponding set.
The transitions between states are controlled by the transition distribution $\mathcal{T}: \mathcal{S} \times \mathcal{A}_0 \times \mathcal{A}_1 \times \ldots \times \mathcal{A}_{K-1} \longmapsto \Delta(\mathcal{S})$.
Each agent $i$ acts according to a parameterized policy $\pi_{\theta_i}: \mathcal{S} \mapsto \Delta(\mathcal{A}_i)$, and the joint policy $\bm{\pi}_\theta = [ \pi_{\theta_1},  \cdots, \pi_{\theta_K} ]$ is the vector of the individual agent policies.
Occasionally, the policy parameters $\theta$ are omitted for convenience.
Note that each agent observes the full state. 
We use subscript ${i-}$ to refer to all agents except $i$, i.e., $\pi_{i-}$  refers to the agent policies, $\{\pi_0, \ldots, \pi_{i-1}, \pi_{i+1}, \ldots, \pi_{K-1}\}$.

Each agent $i$ is also associated with a reward function $R_i: \mathcal{S} \times \mathcal{A}_0 \times \ldots \times \mathcal{A}_{K-1} \longmapsto \mathbb{R}$.
The agent aims to maximize its expected return $ \mathbb{E}_{\bm{\pi}}[\sum_{t=0}^{\infty} \gamma^t r_{i, t}]$, where $r_{i, t}$ is the reward received by agent $i$ at time step $t$, and the discount factor $\gamma \in [0, 1)$ specifies how much to discount future rewards.
In the cooperative tasks considered by this paper, the task rewards are identical across agents.

In Markov games, the optimal policy of an agent depends on the policies of the other agents.
The \textit{best response} policy is the best policy an agent can adopt, given the other agent's policies $\pi_i^* = \argmax_{\pi_i} \mathbb{E}_{\pi_i, \pi_{i-}}[\sum_{t=0}^{\infty} \gamma^t r_{i, t}]$.
If no agent can unilaterally change its policy without reducing its return, then the policies are considered to be in a \textit{Nash equilibrium}.
That is, $\forall i \in [0, K-1], \forall \hat{\pi}_i \neq \pi_i, \mathbb{E}_{\pi_i, \pi_{i-}}[\sum_{t=0}^{\infty} \gamma^t r_{i, t}] \geq \mathbb{E}_{\hat{\pi}_i, \pi_{i-}}[\sum_{t=0}^{\infty} \gamma^t r_{i, t}] $.

The theoretical analysis in Section \ref{section:theory} deals with the above fully observable setting, and assumes a discrete and finite state and action space.
However, the experiments are conducted in partially observable MDPs (POMDPs) with continuous states, which can be formalized as Dec-POMDPs in the multi-agent setting \citep{Oliehoek2012DecPOMDP}.
Dec-POMDPs include two additional elements: the set of observations $\Omega$ and each agent's observation function $O_i: \mathcal{S} \longmapsto \Delta(\Omega)$.

\paragraph{Distribution matching and imitation learning:}
Imitation learning \citep{bakker1996robot, ross2011reduction, schaal1997learning} is a problem setting where an agent tries to mimic trajectories $\{\xi_0, \xi_1, \ldots\}$ where each trajectory $\xi = \{(s_0, a_0), (s_1, a_1), \ldots\}$ is demonstrated by an expert policy $\pi_{E}$. 
Various methods have been proposed to address the imitation learning problem.
Behavioral cloning \citep{bain1995framework} 
applies supervised learning to expert demonstrations to recover the maximum likelihood policy.
Inverse reinforcement learning (IRL) \citep{ng2000algorithms} recovers a reward function which can then be used to learn the expert policy using reinforcement learning.
To do so, $\mathtt{IRL}(\pi_E)$ aims to recover a reward function under which the trajectories demonstrated by $\pi_E$ are optimal.

\citet{ho2016generative} formulate imitation learning as a distribution matching problem and propose the \textsc{gail} algorithm. 
Let the state-action visitation distribution of a joint policy $\bm{\pi} = \langle \pi_1, \ldots \pi_K \rangle$ be:
\begin{equation*} \label{eq:joint-visitation}
    \rho_{\bm{\pi}}(s, \bm{a}) := (1 - \gamma) \prod_{i=1}^K \pi_i(a_i|s) \sum_{t=0}^\infty \gamma^t P(s_t=s | \bm(\pi)).
\end{equation*}
In a multi-agent setting, for agent $i$, 
\begin{equation*}
    \rho_{\pi_i, \pi_{i-}}(s, a) := (1 - \gamma) \pi_i(a|s) \sum_{t=0}^\infty \gamma^t p(s_t=s|\pi_i, \pi_{i-})
\end{equation*}
refers to the marginal state-action visitation distribution of  agent $i$'s policy $\pi_i$, given the other agents' policies $\pi_{i-}$.
In the single agent setting, a policy that minimizes the mismatch of its state-action visitation distribution to the one induced by the expert's trajectories and maximizes its causal entropy $H(\pi)$ is a solution to the $\mathtt{RL} \circ \mathtt{IRL}( \pi_E)$ problem \citep{ho2016generative}.
That is, distribution matching is a solution to the imitation learning problem.

\citet{Guan2021WhenWG} showed that in the single-agent case, the \textsc{gail} algorithm  converges to the expert policy under a variety of policy gradient techniques, including TRPO \cite{schulman2015trpo}.
Let $r_\phi$ be a reward function (based on a discriminator) parameterized by $\phi$, and let $\psi(\phi)$ be a convex regularizer.
\citet{Guan2021WhenWG} formulate the \textsc{gail} problem as the following min-max problem: 
\begin{align}
    \min_{\theta} \max_{\phi} \mathcal{L}(\theta, \phi) \label{eqn:gail_obj} \\
    \text{s.t.  } \mathcal{L}(\theta, \phi) &:=  V({\pi_E}, r_\phi) - V({\pi_\theta}, r_\phi) - \psi(\phi) \nonumber
\end{align}
where $V(\pi, r) = \E_{s_0 \sim \rho_0} \E_\pi [\sum_{t=0}^\infty \gamma^t r_{i, t}]$ is the expected return from some start state when following policy $\pi$ and using reward function $r$.

In the multi-agent setting, imitation learning has the added complexity that the expert trajectories are generated by the interaction of multiple expert policies $\langle \pi_{E_0}, \ldots, \pi_{E_K} \rangle$.
Successful imitation in this setting thus involves the coordination of all $K$ agents' policies.

\section{Theoretical Analysis}
\label{section:theory}

This section provides theoretical grounding for the core proposition of this paper.
The target distribution is assumed to be the empirical distribution of demonstrations from a set of ``expert'' agents in order to ensure that it is achievable by the agents.
Under the conditions stated below, the analysis shows that if $K$ agents independently minimize the distribution mismatch to their respective demonstrations in a turn-by-turn fashion, then agent policies will converge to the joint expert policy. 

The three conditions are as follows.
First, this joint expert policy needs to be coordinated,\footnote{Condition is made concrete in Section \ref{section:architecture}.} but does not have to be a Nash equilibrium with respect to any particular task.
Second, for every policy considered, there is a minimal probability of visiting each state.
Third, each agent learns via a single-agent imitation learning algorithm such that it improves its distribution matching reward at each step.

Next, we establish that if the agents are learning to maximize the mixture of an extrinsic task reward and a distribution matching reward, then the agent policies will converge to a Nash equilibrium with respect to the joint reward.

\subsection{Convergence of Independent \textsc{gail} Learners}

This analysis considers the setting where each agent $i$ performs independent learning updates according to the \textsc{gail} algorithm, to match the visitation distribution of the $i^\text{th}$ expert.
It proposes a condition on an individual agent's \textsc{gail} objective improvement.
If this condition is satisfied, it shows that a lower bound on a joint distribution matching objective is improved.
Further, the lower bound objective converges, demonstrating the convergence of independent \textsc{gail}.

Let the parameterized (discriminator) reward of agent $i$ be $r_{\phi_i}: \mathcal{S} \times \mathcal{A}_i \rightarrow \mathbb{R}$, for $\phi_i \in \Phi$.
At each agent's update, all the other agent policies are held fixed, and the corresponding discriminator has converged to $r_{\phi_i}^{opt}$, where $\phi_i^{opt} \in \argmax_{\phi_i \in \Phi}\mathcal{L}(\theta_i, \phi_i | \theta_{i-})$. \citet{Guan2021WhenWG} showed that the learning process of a single agent repeatedly updating $\phi_i$ converges to $\phi_i^{opt}$.
The update scheme we consider for theoretical purposes is specified in Algorithm \ref{alg_theory}, located in Appendix \ref{app:proofs}. \footnote{A technical appendix is included in the arXiv version of this paper, \texttt{https://arxiv.org/abs/2206.00233}.}

Define the per-agent \textsc{gail} loss as follows: 
\begin{align*}
    \mathcal{L}(\theta_i, \phi_i | \theta_{i-}) &:=  V({\pi_{E_i}}, r^{opt}_{\phi_i} | \pi_{E_{i-}}) \\& - V({\pi_{\theta_i}}, r^{opt}_{\phi_i} | \pi_{\theta_{i-}}) - \psi({\phi_i})  \\
    \text{s.t.    } V(\pi_i, r^{opt}_{\phi_i} | \pi_{i-}) &:= 
    \frac{1}{1 - \gamma} \E_{s \sim \rho_{\bm{\pi}}, a_i \sim \pi_i} \left[r^{opt}_{\phi_i}(s, a_i) \right]
\end{align*}
where $\rho_{\bm{\pi}}(s) = (1 - \gamma)\sum_{t=0}^\infty \gamma^t P(s_t=s | \bm{\pi}) $ is the discounted \textit{state} visitation distribution.

Consider the random variable that is the indicator function $\mathds{1}_{a_j = a^E_j}(s)$ for the event that at state $s$, agent $j$ would take an action that matched expert $j$'s action. Note that the expectation of this indicator is the probability of matching the expert's action\footnote{For the purpose of exposition, assume that the expert policy is deterministic. The theory in this section can be extended to the case where $\bm{\pi}_E$ is stochastic by comparing the distributions over actions.}.
Define the \textbf{joint action-matching objective} as the probability that agent actions match their corresponding experts (plus a constant), weighted by the probability of visiting states:
\begin{equation}
\label{eqn:objective}
J(\bm{\pi}) = \sum_{s \in S} \rho_{\pi}(s) \left[ (K-1) + \E_{\bm{a} \sim \bm{\pi}}  [\mathds{1}_{\bigcap_i a_i = a^E_i}(s)]\right].
\end{equation}
where $\mathds{1}_{\bigcap_i a_i = a^E_i}(s)$ indicates the event that \emph{all} agents take actions that match their corresponding experts. Maximizing $J(\bm{\pi})$ precisely corresponds to solving the multi-agent imitation learning problem because the joint expert policy $\pi_E$ is the unique maximizer of $J(\bm{\pi})$ (Lemma \ref{lemma:jpi_maximizer}, Appendix \ref{app:proofs}).

\begin{restatable}[action-matching objective]{theorem}{thmamobjective}
\label{theorem_amo}
The joint action-matching objective $J(\bm{\pi})$ is lower bounded by the following sum over individual action-matching rewards $\mathds{1}_{a_i = a_i^E(s)}$:
\begin{equation}
    L(\bm{\pi}) :=  \sum_{s \in S} \rho_\pi(s) \left[ \sum_{i=1}^K  \E_{a_i \sim \pi_i} [\mathds{1}_{a_i = a_i^E}(s)] \right].
\end{equation}
\end{restatable}
When an agent updates its policy to optimize its component of $L(\bm{\pi})$, the state visitation distribution $\rho$ might change such that the expected action rewards for other agents decrease.
The next corollary introduces a lower bound on $L(\bm{\pi})$ that is independent of the state visitation distribution $\rho$.
\begin{restatable}[lower bound]{corollary}{corlwrbound}
\label{corr_convergence}
Let $\epsilon$ be the minimum probability of visiting any state. For all $\rho$, $L(\bm{\pi})$ is lower bounded by $L_\epsilon(\bm{\pi})$:
\begin{equation}
    L(\bm{\pi}) > \sum_{s \in S} \sum_{i=1}^K \epsilon \left[ \E_{a_i \sim \pi_i} [\mathds{1}_{a_i = a_i^E}(s) ] \right] =: L_\epsilon({\bm{\pi}}).
\end{equation}
\end{restatable}

With the lower bound, $L_\epsilon$ we make the following assumption to relate our action-matching reward to the GAIL discriminator reward that is improved by the GAIL algorithm.

\begin{restatable}[action-matching reward]{assumption}{assumpaction}
\label{assump:action_match}
For all agents $i$ and all states $s$, an increase in the expected converged GAIL discriminator reward implies an increase to the expected action-matching reward function:
\begin{align*}
     \E_{a_i \sim \pi_i^{t+1}}[ r_{\phi_i}(s, a_i) ] &>  
    \E_{ a_i \sim \pi_i^t}  [  r_{\phi_i}(s, a_i) ] \\ 
    \implies \E_{a_i \sim \pi_i^{t+1}}[ \mathds{1}_{a_i = a^E_i}(s) ] &>  
    \E_{a_i \sim \pi_i^t}  [\mathds{1}_{a_i = a^E_i}(s)].
\end{align*}
\end{restatable}

If the reward $r_{\phi_i}(s, a) = - \log D(s, a)$, as it is in \textsc{gail}, then the assumption above is valid (see Appendix \ref{app:proofs}).

Assumption \ref{assump:action_match} ensures each agent updating its policy leads to improvement in $L_\epsilon(\bm{\pi})$.
This $L_\epsilon(\bm{\pi})$ is a lower bound on the actual objective of interest $J(\bm{\pi})$---by Theorem \ref{theorem_amo} and Corollary \ref{corr_convergence}.
Further, $L_\epsilon(\bm{\pi})$ has a unique global maximizer, which is  $\bm{\pi} = \bm{\pi}_E$ (Lemma \ref{lem:maximizer}).
Thus, while the action reward for the other agents $r_{\phi_j}$ might decrease in the short term, the joint action matching objective across all agents will increase as the learning process continues.
Since $J(\bm{\pi})$ is bounded from above (Lemma \ref{lemma:jpi_bounded}), this process of improving the lower bound will converge to the optimal policy for this objective---the joint expert policy.

\begin{restatable}[convergence]{theorem}{thmconvergence}
\label{theorem_convergence}
Each agent maximizing its individual return over the individual action rewards $r_{\phi_i}$ will converge to the joint expert policy $\bm{\pi}_E$.
\end{restatable} 

\subsection{Multi-agent Learning with Mixed Task and Imitation Reward}

Lemma \ref{lemma:jpi_maximizer} in Appendix \ref{app:proofs} shows that the joint expert policy uniquely maximizes the joint imitation learning objective.
Let $\phi_i^{opt, E_i}$ be the optimal discriminator parameters for the $i^\text{th}$ expert, $\pi_{E_i}$.
From a game theoretic perspective, this lemma implies that these expert policies are a Nash equilibrium for the imitating agents with respect to $r^{opt}_{\phi_i, E_i}$.

Next, note that in imitation learning, it is typically not necessary for the agents to know what the demonstration  actor's task reward is. 
However, suppose that the agents have access to both demonstrations from  policies optimal at task $T$, and the corresponding reward function $R_T$.

Let $R_{I, i} = r^{opt}_{\phi_i, E_i}$, and let the expert policies maximize $R_T$. The expert policies that maximize $R_T$ are in a Nash equilibrium with respect to $R_T$. Theorem \ref{theorem_mixed} states that if the agents are trained to maximize a reward function that is a linear combination of the task reward $R_T$ and $R_{I, i}$, then the converged agent policies are also in a Nash equilibrium with respect to $R_T$. 

\begin{restatable}[]{theorem}{thmmixed}
\label{theorem_mixed}
Let $R_T$ be the reward function used to train the expert policies $\bm{\pi}_{E}$, and let the expert policies have converged with respect to $R_T$ (i.e., they are in a Nash equilibrium with respect to reward $R_T$).
Then $\bm{\pi}_{E}$ are a Nash equilibrium for reward functions of the form, $\alpha R_T + \beta R_{I, i}$, for any  $\alpha, \beta > 0 $.
\end{restatable}

Theorem \ref{theorem_convergence} does not require that the demonstrations originate from optimal policies for some task.
However, Theorem \ref{theorem_mixed} implies that if the demonstrations \emph{do} maximize the reward of a desired task, then the task reward and distribution matching reward can be combined to optimize the same task. 
The proposed algorithm, \textsc{dm$^2$}, takes this approach.

\section{Methods}
\label{section:architecture}

This section discusses practical considerations of fully distributed multi-agent distribution matching, and proposes \textsc{dm$^2$}, an algorithm whose performance is analyzed in Section \ref{section:experiments}.

\subsection{Generating Expert Demonstrations}

Section \ref{section:theory} shows that agents individually following demonstrations from an existing joint policy can converge to said joint policy without centralized training or communication.
In practice, these demonstrations should imply an \emph{achievable} joint expert policy.

For illustration, consider a four tile gridworld, where only one agent is allowed on a tile at a time. Let one of the tiles be labelled, ``A". 
Suppose there are two agents, and each agent $i$ is provided with a separate target state-action distribution, consisting of agent $i$ occupying a tile ``A", and the other agent occupying one of the three remaining tiles.
If both agents simultaneously attempt to match their provided target distributions, then both agents will attempt to occupy tile ``A".
With these demonstrations, it is impossible for both agents to fully match their desired distributions.

The example above shows that for each agent to completely match its desired distribution, the state-action distributions for all agents must be compatible in some way.
This notion of compatibility is defined below.

\begin{definition}[Compatible demonstrations]
\label{def:compatible}
 State-action visitation distributions $\rho_{\pi_i, \hat{\pi}_{i-}}$ from a collection of $K$ policies $\{\pi_i\}_{i=1}^K$ (where $\hat{\pi}_{i-}$ are the other agent policies executed with $\pi_i$ to obtain the state-action visitation distribution $\rho_{\pi_i, \hat{\pi}_{i-}}$) are \emph{compatible} if for all $i$, $s \in \mathcal{S}, a \in \mathcal{A}$, there exists a joint policy $\bm{\pi}' = \langle \pi'_1, \ldots, \pi'_K \rangle$  with the joint state-action visitation distribution $\rho_{\bm{\pi}}(s, \bm{a})$ (Equation \ref{eq:joint-visitation}) such that the marginal state-action visitation distribution for agent $i$ is: 
\begin{align*}
     \rho_{\pi'_i, \pi'_{i-}}(s, a) &:= (1 - \gamma) \pi'_i(a | s)  \sum_{t=0}^\infty \gamma^t P(s_t=s | \pi'_i, \pi'_{i-}) \\
     &\;= \rho_{\pi_i, \hat{\pi}_{i-}}(s, a).
\end{align*}
\end{definition}

Observe that $K$ expert policies that are trained in the same environment to perform a task induce compatible individual state-action visitation distributions, providing a practical method  to obtain compatible demonstrations. 

\subsection{Practical Multi-Agent Distribution Matching}

\textsc{dm$^2$} is inspired by the theoretical analysis in Section \ref{section:theory}, and balances the individual objective of distribution matching with the shared task.
To do so, the agents are provided a mixed reward: part cost function for minimizing individual distribution mismatch, part environment reward.
This approach has been shown to be effective in balancing individual preferences with shared objectives in MARL \citep{durugkar2020balancing, cui2021scalable}.
The individual agent policies are learned by independently updating each agent's policy using an on-policy RL algorithm of choice.
The experiments here use \textsc{ppo} \citep{schulman2017proximal}.

In the experiments, the demonstrations used as targets for the distribution matching are compatible, state-only trajectories---i.e., originating from policies trained jointly on the task of interest.
The use of state-only demonstrations enables learning purely from observations of other agents (e.g. online videos), and is supported by research in the subarea of imitation from observation alone \citep{torabi_generative_2019}. Our experiments also validate the effectiveness of the approach in this setting.
In Section \ref{section:experiments}, we show that the demonstrator policies may possess intermediate competency in the task at hand, and that the demonstrations do not need to be jointly sampled for all the agents.
The proposed learning scheme for training individual agents is summarized by Algorithm \ref{alg} and Figure \ref{fig:method}.

\begin{algorithm2e}[t]
\caption{\textsc{dm$^2$} (\textbf{D}ecentralized \textbf{M}ARL via \textbf{d}istribution \textbf{m}atching)}
\label{alg}
\SetAlgoLined
\KwIn{Number of agents $K$, expert demonstrations $\mathcal{D}_0, \ldots, \mathcal{D}_K$, environment $env$, number of epochs $N$, number of time-steps per epoch $M$, reward mixture coefficient $c$}
\For{$k=0, \ldots, K-1$}{
Initialize discriminator parameters $\phi_k$\;
Initialize policy parameters $\theta_k$\;
}
\For{$n = 0, 1, \ldots , N - 1$}{
Gather $m=1, \ldots, M$ steps of data $(s^m, \bm{a}^m, r^m_{env})$ from $env$\;
\For{$k=0, \ldots, K-1$}{
Sample $M$ states from demonstration $\mathcal{D}_k$\;
Update discriminator $D_{\phi}^k$\;
Compute \textsc{gail} reward $r^m_{k, \textsc{gail}} = - \log D_{k, \phi}(s^m)$ for all $M$ demonstration states\;
Set agent reward $r^m_{k, mix} = r^m_{env} + r^m_{k, \textsc{gail}} * c$\;
Update agent policy $\pi_{\theta}^k$ with data $(s_m, \bm{a}_m, r^m_{k, mix})$ for $m=1, \ldots, M$\;
}
}
\KwOut{$K$ agent policies $\bm{\pi}_\theta$}
\end{algorithm2e}

\section{Experimental Evaluation}
\label{section:experiments}
\begin{figure*}[t!]
\centering
    \begin{subfigure}[]{0.29\linewidth}
        \centering
        \label{subfig:5v6_core_results}
        \includegraphics[width=\linewidth]{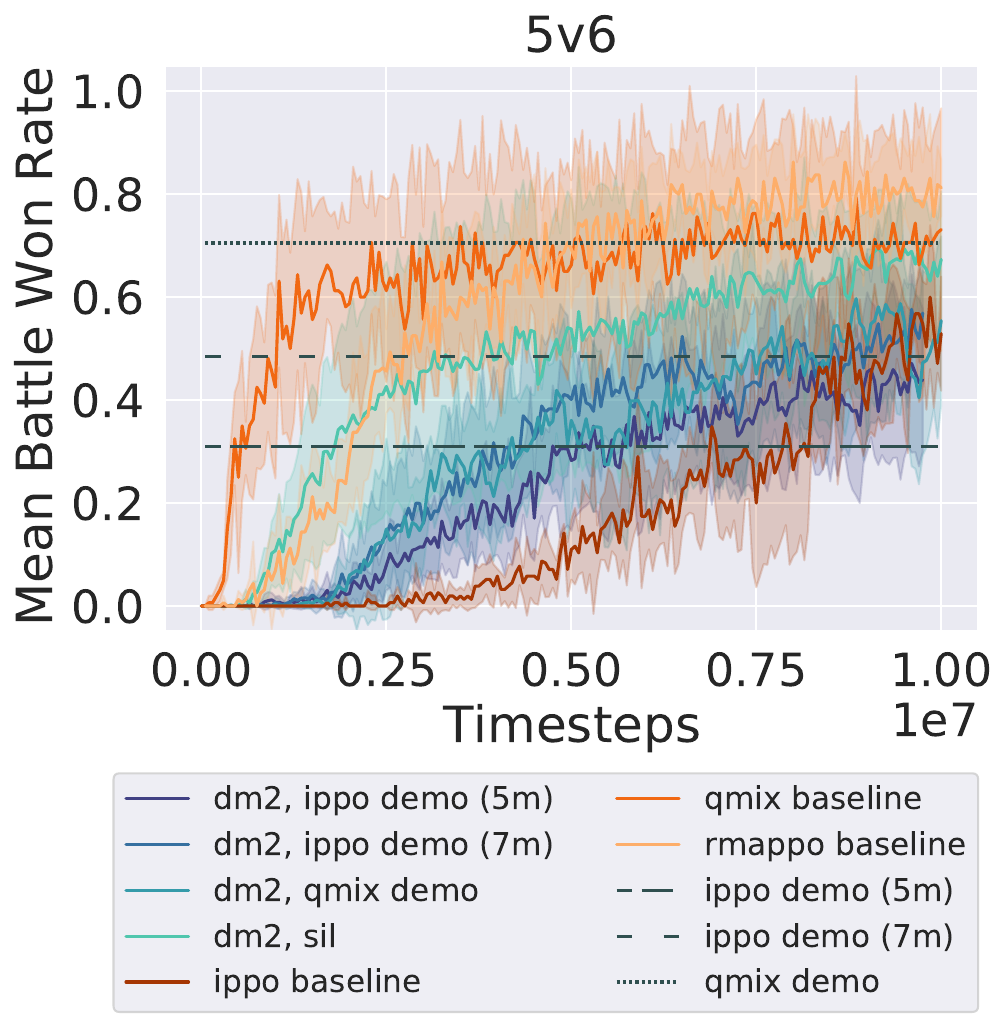}
    \end{subfigure} 
    \hspace{0.025cm}
    \begin{subfigure}[]{0.29\linewidth}
        \centering
        \label{subfig:3sv4z_core_results}
        \includegraphics[width=\linewidth]{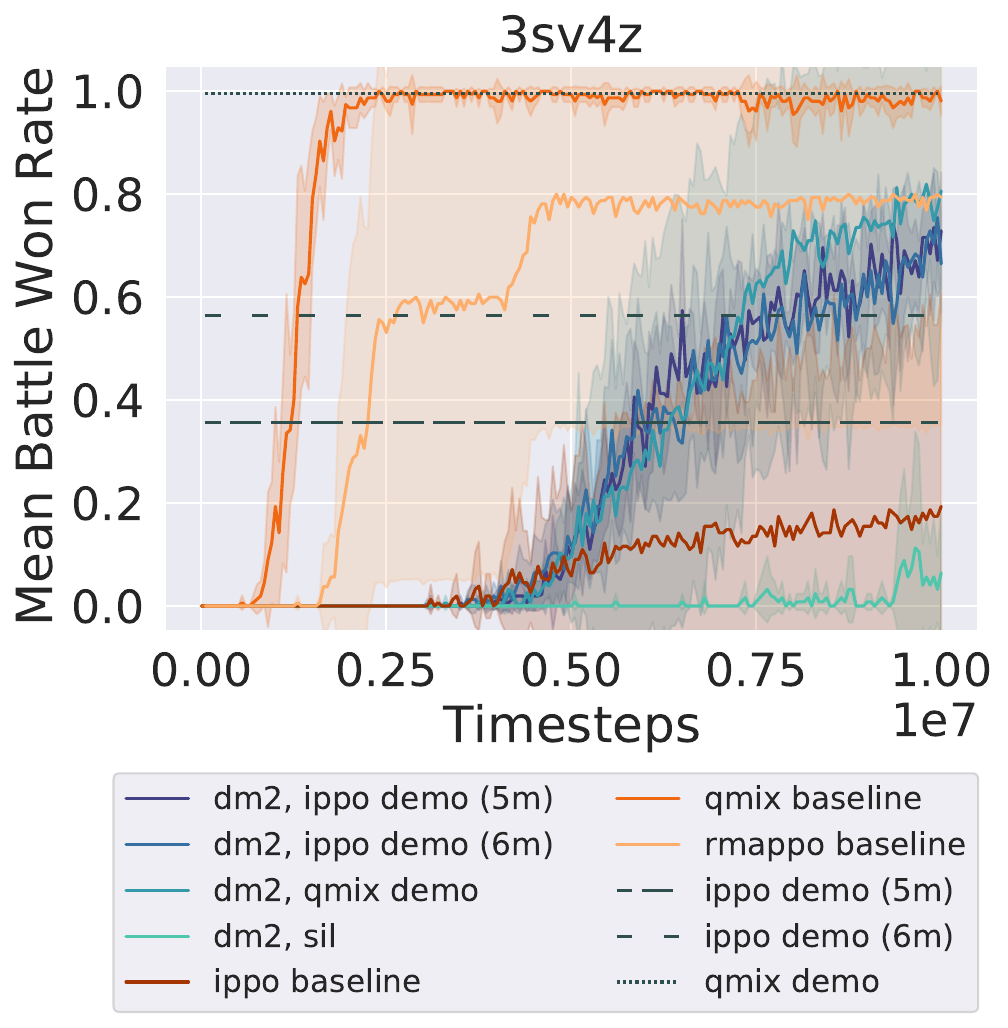}
    \end{subfigure}
    \hspace{0.025cm}
    \begin{subfigure}[]{0.30\linewidth}
        \centering
        \label{subfig:3sv3z_core_results}
        \includegraphics[width=\linewidth]{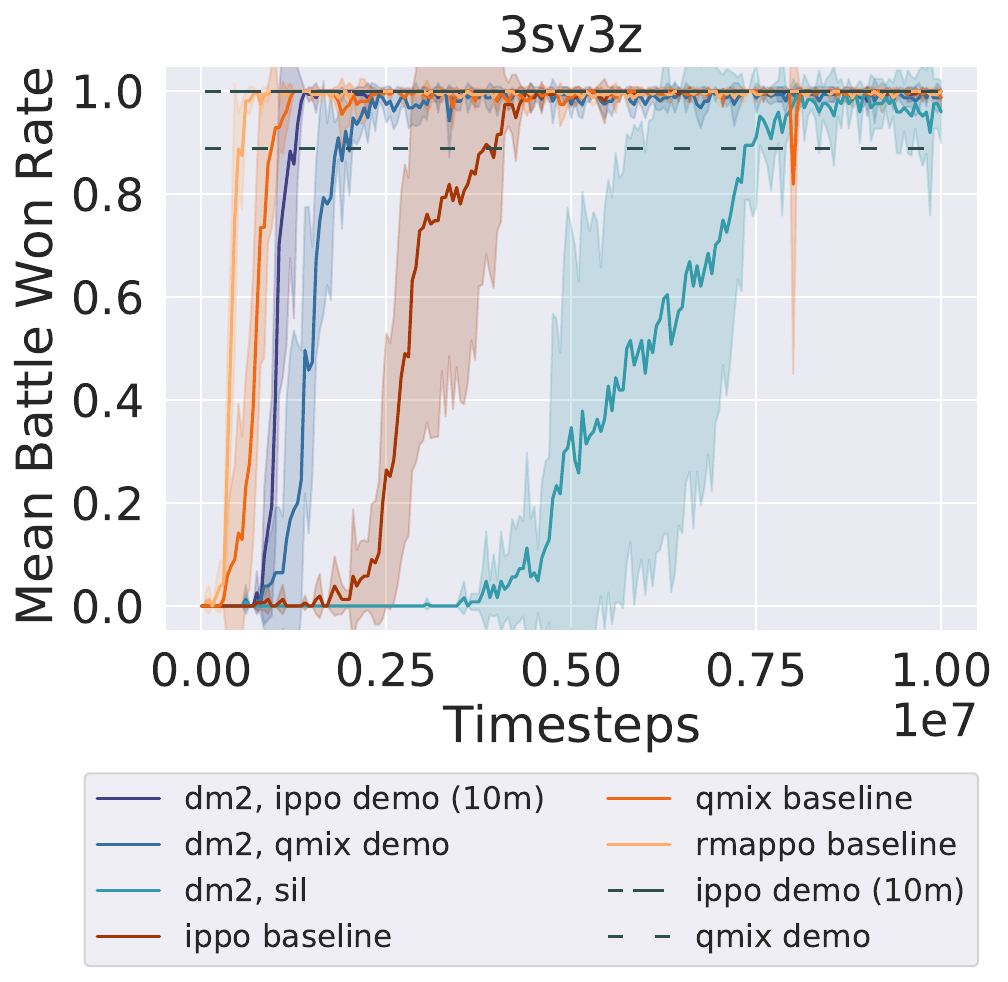}
    \end{subfigure}
\caption{Learning curves of \textsc{dm$^2$} (ours), compared to \textsc{ippo}, \textsc{rmappo}, and \textsc{qmix} baselines on the 5v6, 3sv4z, and 3s3z tasks. \textsc{dm$^2$} is trained with demonstrations from \textsc{ippo} and \textsc{qmix} experts. IPPO demonstrations are sampled from varying points in the training process of the IPPO experts, and therefore vary in quality. A variation with \textsc{sil} is also shown.}
\label{fig:core_results}
\end{figure*}

This section presents two main experiments. The first experiment evaluates whether \textsc{dm$^2$} may improve coordination--and therefore efficiency of learning--over a decentralized MARL baseline.
A comparison against CTDE algorithms is also performed.
The second experiment is an ablation study on the demonstrations that are provided to \textsc{dm$^2$}.
These ablations seek to answer the question whether the faster, improved learning above is due to coordination between experts, or due to individual experts being competent.

Additional experiments in Appendix \ref{app:exp_results} evaluate the effect of the demonstration quality on learning, analyze the usage of demonstrations for behavioral cloning \citep{bain1995framework} instead of distribution matching, and examine the impact of using only distribution matching \textsc{gail} rewards for learning instead of mixing them with the task rewards. The code is provided at \texttt{https://github.com/carolinewang01/dm2}.

\paragraph{Environments:}
Experiments were conducted on the StarCraft Multi-Agent Challenge domain \citep{samvelyan19smac}.
It features cooperative tasks where a team of controllable allied agents must defeat a team of enemy agents.
The enemy agents are controlled by a fixed AI.
The battle is won and the episode terminates if the allies can defeat all enemy agents.
The allies each receive a team reward every time an enemy agent is killed, and when the battle is won. 
StarCraft is a partially observable domain, where an allied agent can observe features about itself, as well as allies and enemies within a fixed radius. 
The specific StarCraft tasks used here (with two additional tasks in Appendix \ref{app:exp_results}) are:
\begin{itemize}[noitemsep,topsep=0pt,leftmargin=*]
    \item 5v6: 5 Marines (allies) and 6 Marines (enemies)
    \item 3sv4z: 3 Stalkers (allies) and 4 Zealots (enemies)
    \item 3sv3z: 3 Stalkers (allies) and 3 Zealots (enemies)
\end{itemize}
\paragraph{Baselines:} 
\textsc{dm$^2$} is compared against a naive decentralized MARL algorithm, independent \textsc{ppo} \citep{schulman2017proximal} (\textsc{ippo}), where individual \textsc{ppo} agents directly receive the team environment reward. Although agents trained under the \textsc{ippo} scheme cannot share information and see only local observations, prior work has shown that \textsc{ippo} can be surprisingly competitive with CTDE methods \citep{yu2021surprising}. We also compare against two  widely used CTDE methods, \textsc{qmix} \citep{rashid18qmix} and \textsc{rmappo} \citep{yu2021surprising}.
These CTDE methods have the advantage of a shared critic network that receives the global state during training.
Thus, their performance is expected to be better than that of decentralized methods with no communication. 
\paragraph{Setup:}
\textsc{dm$^2$} uses the same \textsc{ippo} implementation as the baseline, with the addition of a \textsc{gail} discriminator for each independent agent $i$ to generate an imitation reward signal, $r_{i, \textsc{gail}}$. The scaled \textsc{gail} reward is added to the environment reward $r_{env}$, with scaling coefficient $c \in \mathbb{R}$: $r_{i, mix} = r_{env} + r_{i, \textsc{gail}} * c$. Learning curves of all algorithms are the mean of 5 runs executed with independent random seeds, where each run is evaluated for 32 test episodes at regular intervals during training. The shaded regions on the plots show the standard error. The evaluation metric is the mean rate of battles won against enemy teams during test episodes.

The data for the \textsc{gail} discriminator consists of 1000 joint state-only trajectories (no actions). The data is sampled from checkpoints during training runs of baseline \textsc{ippo} with the environment reward, and \textsc{qmix} with the environment reward.
In runs of \textsc{dm$^2$}, each agent imitates the marginal observations of the corresponding agent from the dataset (i.e., agent $i$ will imitate agent $i$'s observations from the dataset) \footnote{The allied agent teams in our experiments have the same state/action spaces. Thus, the mapping of agents to demonstration trajectories does not matter, as long as it is fixed}.
For each task, demonstrations are sampled from \textsc{ippo} and \textsc{qmix}-trained joint expert policies, executed stochastically for \textsc{ippo} and with an $\epsilon$-greedy sampling for \textsc{qmix}. The win rates achieved by the demonstration policies are plotted as horizontal lines on the graphs.
Additionally, \textsc{sil}\cite{hao2019independent} can be seen as a variation of \textsc{dm$^2$}, with the target demonstrations being the agent's most successful prior trajectories.
Experimental details such as hyperparameters are specified in Appendix \ref{app:exp_details}.
\begin{figure*}[tbh]
\centering
    \begin{subfigure}[]{0.31\linewidth}
        \centering
        \includegraphics[width=\linewidth]{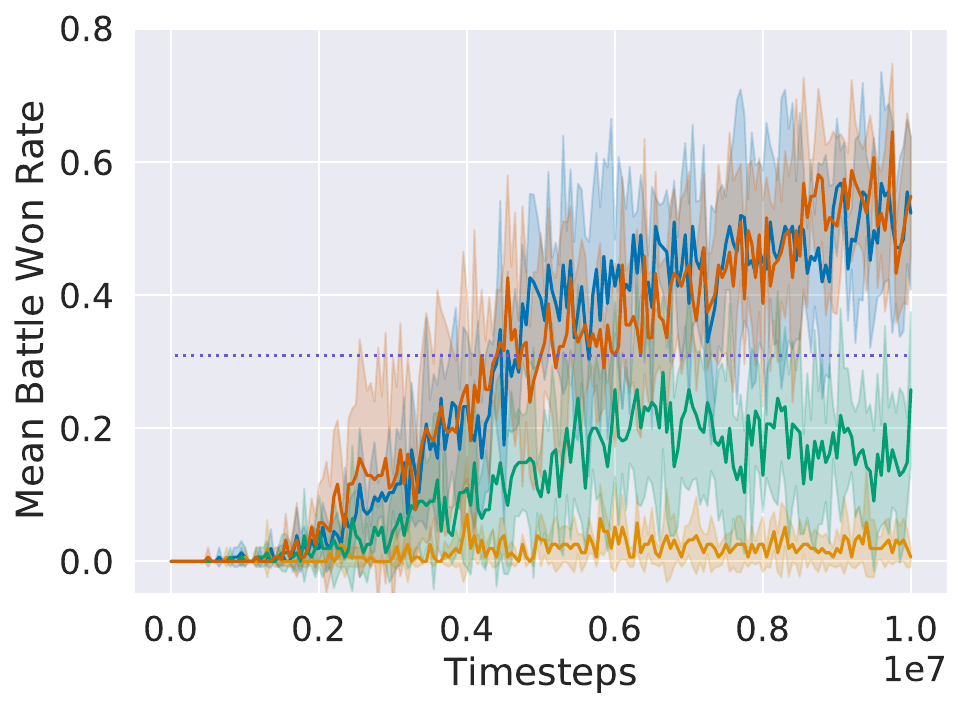}
    \end{subfigure}
    \hspace{0.8cm}
    \begin{subfigure}[]{0.31\linewidth}
        \centering
        \includegraphics[width=\linewidth]{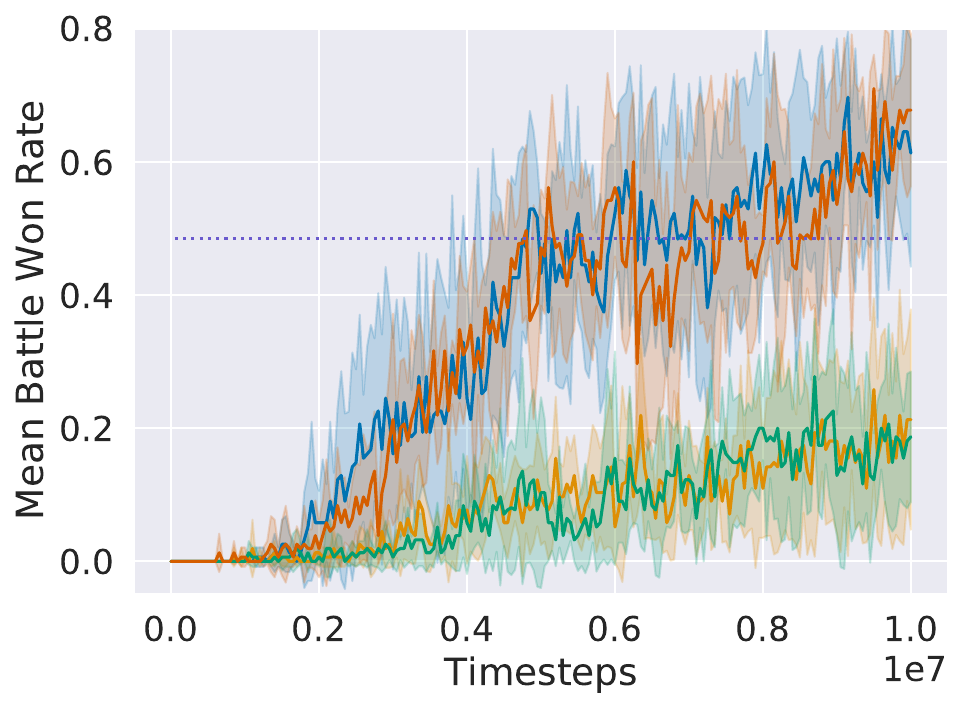}
    \end{subfigure}
    \\
    \begin{subfigure}[]{0.65\linewidth}
        \includegraphics[width=\linewidth]{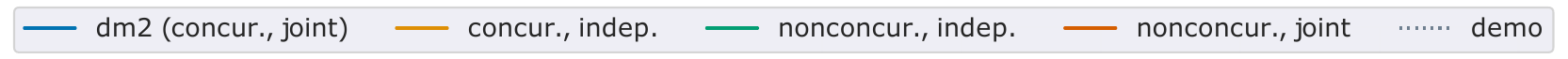}
    \end{subfigure}
\caption{Ablations for \textsc{ippo} trained with $r_{mix}$ on the 5v6 task. The case where the demonstrations are concurrently sampled from co-trained (joint) expert policies corresponds to \textsc{dm$^2$}. Left: Experiments performed with lower quality IPPO demonstration (5m). Right: Experiments performed with higher quality IPPO demonstration (7m).}
\label{fig:5v6_ablations}
\end{figure*}

\subsection{Main Results}

Figure \ref{fig:core_results} shows that in all three tasks, \textsc{dm$^2$} significantly improves learning speed over \textsc{ippo} (the decentralized baseline).
\textsc{qmix} and \textsc{rmappo} (the CTDE baselines) learn faster than \textsc{dm$^2$} and \textsc{ippo} on both tasks, illustrating the challenging nature of the decentralized cooperation problem.
However, on 5v6 and 3sv3z, all methods converge to a similar win rate towards the end of training.
For the demonstrations from \textsc{ippo} experts, \textsc{dm$^2$} surpasses the win rate of the demonstrations.
Despite the significant variance in win rates among the demonstrations for each task, \textsc{dm$^2$} performs similarly.
Similar robustness to demonstration quality is seen even with 10 different demonstration qualities (Figure \ref{fig:demo_learning_curves_all} in Appendix \ref{app:exp_results}). The relative invariance to demonstration quality suggests that the demonstrations provide a useful coordinating signal,  enabling agents to discover higher-return behaviors than those portrayed in the demonstrations. 

On the other hand, using trajectories from earlier in training as demonstrations (\textsc{sil}) shows inconsistent performance.
Its performance is comparable to \textsc{rmappo} in 5v6, but it learns more slowly in the other domains.
This may occur because \textsc{sil} requires examples of successful coordination by the agents, which may be rare in certain tasks.
Using past trajectories as demonstrations also leads to a nonstationary target distribution for imitation, potentially negatively impacting the learning procedure. 

\subsection{Ablation Study}
This section presents an ablation study on the demonstrations.
It investigates whether the coordination of expert agents or their individual competency is more important to the success of \textsc{dm$^2$} in the main experiment.
This comparison is done by considering demonstrations that vary in two dimensions: whether they were sampled from expert teams that form a joint policy (co-trained), or whether they were sampled simultaneously (concurrently sampled).

The first dimension tests the requirement that agents are trained to coordinate with each other, while the second tests whether agents must act  together when generating demonstrations.
Concurrently sampled demonstrations of agents that were not co-trained, gives us examples of individually competent agents acting in the multi-agent setting.

The experiments apply \textsc{dm$^2$} to four possible demonstration styles that vary in the aforementioned two dimensions. A detailed  explanation of how these four demonstration styles were constructed is provided in Appendix \ref{app:exp_details}. The study is performed on the 5v6 task, with the same hyperparameters used in the experiments of the previous section. 

Figure \ref{fig:5v6_ablations} shows the learning curves of the four combinations. The axis that appears to make the greatest difference in learning is whether the demonstrations originate from expert policies that were co-trained, and were thus coordinated.
Whether the agent demonstrations were concurrently sampled does not appear to significantly impact learning.
Similar trends are observed when \textsc{dm$^2$} is trained with the lower quality demonstration (Figure \ref{fig:5v6_ablations}, left).

\section{Discussion and Future Work} \label{section:discussion}

This paper studies distributed MARL for cooperative tasks without communication or explicit coordination mechanisms.
Fully distributed MARL is challenging, since simultaneous updates to different agents' policies can cause them to diverge.
The benefits of distributed MARL are abundant.
Decentralized training could make agents more robust to the presence of agents they were not trained with (e.g. humans). 
Decentralized training could also enable coordination while preserving the privacy of each agent.

The theoretical analysis of this paper shows that individual agents updating their policies turn-by-turn to reduce their distribution mismatch to corresponding expert distributions improves a lower bound to the joint action-matching objective against the joint expert policy. Fully maximizing the lower bound corresponds to recovering the joint expert policy.
The experiments verify that mixing the task reward with the distribution matching reward accelerates cooperative task learning, compared to learning without the distribution matching objective.
The ablation experiments show that expert demonstrations should be from policies that were trained together, but not necessarily concurrently sampled. 

While this work is a meaningful step towards fully distributed multi-agent learning via distribution matching, some open questions remain. 
Future work could consider whether demonstrations sampled from expert policies with other properties, such as those trained with reward signals corresponding to different tasks, could be beneficial for distributed learning. 
The method proposed in this paper could also be leveraged to combine human demonstrations with a task reward for applications of MARL ranging from expert decision making (similar to that done by \citet{gombolay2018robotic} in the context of medical recommendation) or in the context of complex multi-agent traffic navigation \citep{behbahani2019learning}.
Another potential path forward would be considering human in the loop settings such as the TAMER architecture \cite{knox2009interactively}, but in a fully distributed multi-agent setting.

\newpage
\section*{Acknowledgements}

\thanks{This work has taken place in the Learning Agents Research
Group (LARG) at the Artificial Intelligence Laboratory, The University
of Texas at Austin, and at SparkCognition Applied Research.  LARG research is supported in part by the
National Science Foundation (CPS-1739964, IIS-1724157, FAIN-2019844),
the Office of Naval Research (N00014-18-2243), Army Research Office
(W911NF-19-2-0333), DARPA, General Motors, Bosch, and
Good Systems, a research grand challenge at the University of Texas at
Austin.  The views and conclusions contained in this document are
those of the authors alone.  Peter Stone serves as the Executive
Director of Sony AI America and receives financial compensation for
this work.  The terms of this arrangement have been reviewed and
approved by the University of Texas at Austin in accordance with its
policy on objectivity in research.}

%% file: content_appendix.tex
\section*{Appendix}

\section{Convergence Proof Details} \label{app:proofs}

In Section \ref{section:theory}, we lay out some of the conditions for our theoretical analysis.
One of these conditions, that for every policy considered there is a minimal probability of visiting each state, is formalized below.

\begin{condition} 
\label{condition:ergodicity_strong}
Let $0 < \epsilon < 1$, and let $\rho$ be the state visitation distribution induced by any joint policy during training. For all agents $i$ and for all $s$, suppose that $\rho(s) \geq \epsilon$. 
\end{condition}

\begin{algorithm2e}[ht]
\caption{Distributed MARL with distribution matching}
\label{alg_theory}
\SetAlgoLined
\KwIn{Number of agents $K$, expert demonstrations $\mathcal{D}_0, \ldots, \mathcal{D}_K$, environment $env$, number of time-steps per epoch $M$}
\For{$k=0, \ldots, K-1$}{
Initialize discriminator parameters $\phi_k$\;
Initialize policy parameters $\theta_k$\;
}
\While{any($\pi_{\theta_k}$) \text{not converged}}{
\For{$k=0, \ldots, K-1$}{
Gather $m=1, \ldots, M$ steps of data $(s^m, \bm{a}^m, r^m_{env})$ from $env$\;
    Update agent discriminator $r_{\phi_k}$ to maximize Equation \ref{eqn:gail_obj} until convergence to $r_{\phi_k}^{opt}$\;
    Update agent policy $\pi_{\theta_k}$ using TRPO to minimize Equation \ref{eqn:gail_obj}
}

}
\KwOut{$K$ agent policies $\bm{\pi}_\theta$}
\end{algorithm2e}

Recall that the \textbf{joint action-matching objective} is defined over the expected state visitation as the probability that all agent actions match their corresponding experts (plus a constant):

\begin{equation*}
J(\bm{\pi}) = \sum_{s \in S} \rho_{\pi}(s) \left[ (K-1) + \E_{\bm{a} \sim \bm{\pi}}  [\mathds{1}_{\bigcap_i a_i = a^E_i}(s)]\right].
\end{equation*}
where $\mathds{1}_{\bigcap_i a_i = a^E_i}(s)$ indicates the event that \emph{all} the agents took actions that matched their corresponding experts.

We first prove some properties of $J(\bm{\pi})$.

\begin{lemma} \label{lemma:jpi_bounded}
The objective $J(\bm{\pi})$ is bounded by 
\begin{equation*}
   (K-1) \leq  J(\bm{\pi}) \leq K.
\end{equation*}
\end{lemma}

\begin{proof}
As shorthand, define $f(\bm{\pi}) = \sum_{s \in S} \rho_{\pi}(s) \E_{\bm{a} \sim \bm{\pi}} [  \mathds{1}_{\bigcap_i a_i = a^E_i}(s)]$. Then,  
\begin{equation*}
    J(\bm{\pi}) = (K-1) + f(\bm{\pi}).
\end{equation*}
First, note that $f(\bm{\pi}) \geq 0$ because it is a weighted sum of expectations over indicator functions,  where all weights are non-negative, and it is precisely $0$ if for all states, the joint agent policy $\bm{\pi}$ does not match the correct expert actions. Thus, $J(\bm{\pi})$ is lower bounded by $(K-1)$. For the upper bound, notice that the outer summation is equivalent to the expectation under state visitation distribution $\rho_\pi$. Inside, each expectation over the indicators can be at most 1, implying that $f(\bm{\pi})$ is at most $1$. Thus, $J(\bm{\pi}) \leq K$.
\end{proof}

\begin{lemma} \label{lemma:jpi_maximizer} Suppose the joint expert policy $\bm{\pi}_E$ is deterministic. Then $\bm{\pi}_E$ is the unique maximizer of $J(\bm{\pi})$.
\end{lemma}

\begin{proof}
(Maximization) Since $\bm{\pi}_E$  is deterministic, by definition, each term
\begin{equation*}
    \E_{\bm{a} \sim \bm{\pi}_E} [  \mathds{1}_{\bigcap_i a_i = a^E_i}(s)] = 1.
\end{equation*} 
Thus, $J(\bm{\pi}_E) = K$, which means that $\bm{\pi}_E$ achieves the upper bound of $J$. 

(Uniqueness) Suppose there exists another policy $\bm{\pi} \neq \bm{\pi}_E$ that also achieves the upper bound, i.e. $J(\bm{\pi}) =  K$.
Let $a_i^E := \pi_{E_i}(s)$.
Then there must be an agent $i$ such that with positive probability, $a_i \sim \pi_i(s) $ such that $a_i \neq a_i^E$. Then it is immediate that at state $s$, $\E_{\bm{a} \sim \bm{\pi}} [  \mathds{1}_{\bigcap_i a_i = a^E_i}(s)] < 1$.
Combined with the non-zero probability of visiting every state (Condition \ref{condition:ergodicity_strong}), this inequality then implies $J(\bm{\pi}_E) < K$.
By contradiction, $\bm{\pi}_E$ is the unique optimizer for $J$.
\end{proof}

Next, we establish that individual agents performing GAIL updates maximizes a lower bound on $J(\bm{\pi})$.
We leverage the single-agent GAIL convergence result by \cite{Guan2021WhenWG} to show this result.

\begin{lemma} \label{lemma:guan}
Let $t$ be the time step at which agent $i$'s policy is updated.
For all agents $j$, denote the optimal discriminators of Equation \ref{eqn:gail_obj} as $r_{\phi_j}^{opt}$.
Suppose agent $i$ updates its policy parameters from $\theta_i^t$ to $\theta_i^{t+1}$ such that $\mathcal{L}(\theta_i^{t+1}, \phi_i^{opt} | \theta_{i-}^t) < \mathcal{L}(\theta_i^{t}, \phi_i^{opt} | \theta_{i-}^t)$. This decrease in loss is equivalent to increasing the agent $i$'s expected discriminator reward.
\end{lemma}

\begin{proof}
First, note that updating a single agent policy while keeping all discriminators fixed does not alter the expert value term $V(\pi_{E_k}, r_{\phi_k^{opt}})$ or the regularizer term $\psi(\phi_k^{opt})$ in the loss definition $\mathcal{L}$.
Thus, the condition that agent $i$'s loss has decreased is equivalent to the value of agent $i$ increasing: 
\begin{equation} \label{eq:assumption}
    V(\theta_i^{t+1}, \phi_i^{opt} | \theta_{i-}^t) > V(\theta_i^{t}, \phi_i^{opt} | \theta_{i-}^t).
\end{equation}

For convenience of notation, agent $i$'s policy at time t will be written as $\pi_i^t$, and the $i$th discriminator $r^{opt}_{\phi_i}$ implicitly indicated by the action subscript,  $r(s, a_i)$. Similarly, we will write the \textbf{state} visitation distribution induced by the policies $\pi_i^{t+1}, \pi_{i-}^t$  by $\rho^{t+1}$, and the distribution induced by $\pi_i^{t}, \pi_{i-}^t$  as $\rho^t$. 

Rewriting Equation \ref{eq:assumption} in terms of visitation distributions:

\begin{equation*}
    \frac{1}{1-\gamma} \E_{s \sim \rho^{t+1}, a_i \sim \pi_i^{t+1}}[ r_{\phi_i}(s, a_i) ] > \frac{1}{1-\gamma} 
    \E_{s \sim  \rho^{t}, a_i \sim \pi_i^t}  [  r_{\phi_i}(s, a_i) ]
\end{equation*}

\begin{equation} \label{eq:guan}
     \E_{s \sim \rho^{t+1}, a_i \sim \pi_i^{t+1}}[ r_{\phi_i}(s, a_i) ] >  
    \E_{s \sim  \rho^{t}, a_i \sim \pi_i^t}  [  r_{\phi_i}(s, a_i) ].
\end{equation}

\end{proof}

We next show that $J(\bm{\pi})$ is lower bounded by the sum of the individual action-matching rewards for all agents $i$, over all states.

\thmamobjective*

\begin{proof}
Let us begin by rewriting $J(\bm{\pi})$ in terms of action \textit{mismatches}. 

\begin{align*}
J(\bm{\pi}) &= \sum_{s \in S} \rho_\pi(s) \left[ (K-1) + \E_{\bm{a} \sim \bm{\pi}} [ 1 - \mathds{1}_{\bigcup_i a_i \neq a^E_i}(s)] \right]  \\
    &= \sum_{s \in S} \rho_\pi(s) \left[ (K - 1) + 1 + \E_{\bm{a} \sim \bm{\pi}} [ - \mathds{1}_{\bigcup_i a_i \neq a^E_i}(s)] \right].
\end{align*}

This formulation allows us to apply the Union Bound: 
\begin{align*}
    J(\bm{\pi}) &\geq \sum_{s \in S} \rho_\pi(s) \left[ K +  \sum_{i=1}^K  \E_{a_i \sim \pi_i} [ - \mathds{1}_{a_i \neq a^E_i}(s)] \right] \\
&=  \sum_{s \in S} \rho_\pi(s) \left[ \sum_{i=1}^K \E_{a_i \sim \pi_i} [ 1 - \mathds{1}_{a_i \neq a^E_i}(s)] \right]\\ 
& = L(\bm{\pi}).
\end{align*}
\end{proof}

Theorem \ref{theorem_amo} relates the multi-agent imitation learning objective to a sum over single-agent imitation learning objectives. 
This is important because each agent updates independently to improves its own learning objective in our setting. 
Note also that the additive form of $L(\bm{\pi})$ is similar to the value factorization assumptions made by algorithms like VDN \citep{sunehag18vdn}.

It is difficult to say anything directly about the expected action-matching reward of agent $j \neq i$, $\E_{s \sim \rho^{t+1}} [\mathds{1}_{a_j = a_j}^E(s)]$, as $\rho^{t+1}$ may be a state distribution over which agent $j$ makes more mistakes (i.e. taking actions that don't match the expert's).
While in general agent $j$'s expected reward may decrease due to agent $i$'s update, we show that under Assumptions \ref{condition:ergodicity_strong} and \ref{assump:action_match}, agent $i$'s update increases a lower bound to $L(\bm{\pi})$ that is independent of the state distribution.

\corlwrbound*
\begin{proof}
The proof follows from the definition of $L(\bm{\pi})$ and that for all $s \in S$, and all $\rho$ encountered in training, $\rho(s) > \epsilon$.
\end{proof}

\begin{observation}
$L_\epsilon(\bm{\pi})$ is bounded by
\begin{equation*}
    0 \leq L_\epsilon(\bm{\pi}) \leq \epsilon |\mathcal{S}| K.
\end{equation*}
\end{observation}

\begin{lemma} \label{lem:maximizer}
Suppose the joint expert policy $\bm{\pi}_E$ is deterministic. Then $\bm{\pi}_E$ is the unique maximizer of $L_{\epsilon}(\bm{\pi})$.
\end{lemma}
\begin{proof}
Proof for this Lemma follows closely the proof for Lemma \ref{lemma:jpi_maximizer}.

(Maximization) Since $\bm{\pi}_E$  is deterministic, by definition, for each agent $i$, each term
\begin{align*}
    \E_{a_i \sim \pi_{E_i}} [  \mathds{1}_{a_i \neq a^E_i}(s)] &= 0 \\
    \E_{a_i \sim \pi_{E_i}} [1 -  \mathds{1}_{a_i \neq a^E_i}(s)] &= 1.
\end{align*} 
Summing over all agents and taking the sum weighted by $\epsilon$ over all states, $L_\epsilon(\bm{\pi}_E) = \epsilon |\mathcal{S}|K$, which means that $\bm{\pi}_E$ achieves the upper bound of $L_\epsilon$. 

(Uniqueness) Suppose there is at least one agent policy $\pi_i \neq \pi_{E_i}$ such that the joint policy $\bm{\pi}$ also achieves the upper bound, i.e. $L_\epsilon(\bm{\pi}) = \epsilon |\mathcal{S}|K$.
Then there must be a state $s$ such that with non-zero probability, $a_i \sim \pi_i(s)$ such that $a_i \neq a^E_i$.
It follows that $\E_{a_i \sim \pi_{E_i}} [  \mathds{1}_{a_i \neq a^E_i}(s)] > 0$, meaning $L_\epsilon(\bm{\pi}_E) < \epsilon |\mathcal{S}|K$.
By contradiction, $\bm{\pi}_E$ is the unique optimizer of $L_\epsilon$.
\end{proof}

Lemma \ref{lemma:guan} states that a single agent $i$ updating its policy to improve its own GAIL loss is equivalent to increasing the expected action reward $r_{\phi_i}(s, a)$, where the expectation over states is with respect to some updated state visitation distribution $\rho^{t+1}$ (Equation \ref{eq:guan}).
Next, we make an assumption to relate our action-matching reward to the GAIL discriminator reward that is improved by the GAIL algorithm.

\assumpaction*

This assumption is not as strong as it may first appear. 

First, note that the converged  GAIL reward consists of the negative discriminator prediction, $r_{\phi_i}(s, a_i) = - \log D(s, a_i)$.

The discriminator predicts the likelihood ratio between the target visitation and the mix of target and agent visitation.
\begin{align*}
    r_{\phi_i}(s, a_i) = - \log D(s, a_i) &= - \log \left( \frac{\rho_E(s, a_i)}{\rho_E(s, a_i) + \rho(s, a_i)} \right) \\
    &= - \log \left( \frac{\rho_E(s) * \pi_E(a_i|s)}{\rho_E(s) * \pi_E(a_i|s) + \rho(s) * \pi(a_i|s)} \right) \\
    &\geq - \log \left( \frac{\rho_E(s) * \pi_E(a_i|s)}{\rho_E(s) * \pi_E(a_i|s) + \epsilon * \pi(a_i|s)} \right) \\
    &=: r_\epsilon(s, a_i).
\end{align*}

The minimal state visitation assumption states that $\rho(s) \geq \epsilon$ for all $s$, allowing us to relate $r_{\phi_i}(s, a_i)$ to a state-visitation \textit{independent} reward, $r_\epsilon(s, a_i)$. As we will argue next, $r_\epsilon(s, a_i)$ is a similar quantity to the action-matching indicator reward. 

To see this, first suppose $a_i \neq a_i^E$, which would imply that the action-matching indicator function is 0. Since the expert policy is assumed to be deterministic, $\pi_E(a_i | s) = 0$,  implying that $r_\epsilon(s, a_i) = 0$ as well. If $a_i = a_i^E$, then $r_\epsilon(s, a_i^E)$ is not zero, and the only way for agent $i$ to increase $r_\epsilon(s, a_i^E)$ is to increase $\pi(a_i^E | s)$ --- the probability of matching the expert's action. 
Thus, an increase in 
 $r_{\phi_i}(s, a_i)$ implies an increase in $r_\epsilon(s, a_i)$, which behaves similarly to the action-matching indicator function.

Thus far, we have established that each individual agent's policy improvement under the GAIL reward improves a lower bound to the joint action-matching objective, $J(\bm{\pi})$. The following shows that within finite updates, each agent will be able to independently improve its value function until it converges to the expert policy. 

\begin{condition} \label{condition:single-agent-improvement}
Let $V(\pi)$ denote the value of an agent following a single-agent imitation learning algorithm. $|V(\pi_t) - V(\pi^E)|$ is then the optimality gap at update $t$ of the agent. Suppose that $|V(\pi_t) - V(\pi^E)| \leq \eta(t)$, where as $t \rightarrow \infty$, $\eta(t) \rightarrow 0$. 
\end{condition}

This condition says that the single-agent imitation learning process should converge to the optimal (expert) policy with convergence rate dictated by $\eta(t)$. For our setting, \citet{Guan2021WhenWG} shows that the single-agent GAIL algorithm converges (Theorem 3 and 4).

The next corollary shows that convergence of the single-agent imitation learning process is sufficient to guarantee the convergence of the multi-agent imitation learning scheme discussed in the main paper.

\begin{corollary} \label{corr:lep_improve}
There exists $H \in \mathbb{N}^+,$ $H < \infty$ such that within $H$ updates, agent $i$ is able to improve its policy such that it increases the probability of matching the expert's action, summed over all states:
\begin{equation*}
    \sum_{s\in S} \epsilon \left[ \frac1c \E_{a \sim \pi_i^{t+H}} [ r_{\phi_i}(s, a) ]\right] > \sum_{s\in S} \epsilon \left[ \frac1c \E_{a \sim \pi_i^{t}} [ r_{\phi_i}(s, a) ]\right].  
\end{equation*}
\end{corollary}

\begin{proof}
For a single agent $i$, define $V(\pi) = \sum_{s \in S} \rho_{\pi}(s) \left[ \frac1c \E_{a \sim \pi_i}[r_{\phi_i}(s, a)] \right] $. Rewrite as follows: 

\begin{align*}
    V(\pi) &= \sum_{s \in S} \rho_{\pi}(s) \E_{a \sim \pi_i}[r_{\phi_i}(s, a)] \\ 
    &= \frac1c \sum_{s \in S} (\rho_{\pi}(s) - \epsilon)  \E_{a \sim \pi_i}[r_{\phi_i}(s, a)] + \frac1c \sum_{s\in S} \epsilon \E_{a \sim \pi_i}[r_{\phi_i}(s, a)].
\end{align*}

Define the first term as $a(\pi) := \frac1c \sum_{s \in S} (\rho_{\pi}(s) - \epsilon)  \E_{a \sim \pi_i}[r_{\phi_i}(s, a)]$, and the second term as $b(\pi) := \frac1c \sum_{s\in S} \epsilon \E_{a \sim \pi_i}[r_{\phi_i}(s, a)]$. Note that $b(\pi)$ is the quantity of interest in the corollary. 

The properties of the max operation directly imply that 

\begin{align*}
    \max_{\pi } V(\pi) &= \max_{\pi} [a(\pi)  + b(\pi)] \\
    &\leq \max_{\pi} a(\pi) + \max_{\pi} b(\pi).
\end{align*}

The expert policy $\pi_E$ should maximize $V(\pi)$ for any single-agent imitation learning algorithm. Note also that $\pi_E$ maximizes $a(\pi)$ and $b(\pi)$. Thus in our setting, the inequality in the above set of equations is actually an equality for the expert policy: 
\begin{align*}
    \max_{\pi } V(\pi) &= V(\pi_E) \\
    &= a(\pi_E) + b(\pi_E) 
    = \max_{\pi} a(\pi) + \max_{\pi} b(\pi).
\end{align*}

Assume that there is a policy $\pi \neq \pi_E$ such that $b(\pi) < b(\pi_E)$, and that $b(\pi)$ cannot be improved within finite updates.
This contradicts Condition 
\ref{condition:single-agent-improvement}, which establishes that the single-agent imitation learning algorithm should be able to improve the agent policy until its value converges to the value of the expert policy.
\end{proof}

Corollary $\ref{corr:lep_improve}$ implies that within a finite number of policy updates by agent $i$, the quantity $L_\epsilon(\bm{\pi})$ increases, because the other terms corresponding to agents $j \neq i$ are unchanged. 
By Theorem \ref{theorem_convergence} and Lemma \ref{lem:maximizer}, $L_\epsilon(\bm{\pi})$ is upper bounded by the constant $\epsilon |S|K$. 
Thus, by the monotone convergence theorem, the objective $L_\epsilon(\bm{\pi})$ converges. 
Further, the expert policy $\bm{\pi_E}$ is a maximizer of $L(\bm{\pi})$, $L_\epsilon(\bm{\pi})$, and $J(\bm{\pi})$ (Lemma \ref{lemma:jpi_maximizer} and Lemma \ref{lem:maximizer}).

\subsection{Mixed Task and Imitation Reward}

\thmmixed*
\begin{proof}

Let $R_{c, i} = \alpha R_T + \beta R_{I, i}$. The following reasoning is on a per-agent basis, so we drop the $i$ from $R_{c, i}$ and $R_{I, i}$ for convenience. For $\pi_{E_i}$ to not be a Nash equilibrium with respect to $R_c$ there needs to exist a policy $\tilde{\pi}_{i}$ such that 
\begin{equation*}
\mathbb{E}[R_c(\tilde{\pi}_{i}(s) | \pi_{E_{-i}})] > \mathbb{E}[R_c(\pi_{E_i}(s) | \pi_{E_{-i}})].
\end{equation*}
That implies 
\begin{align*}
    & \alpha \mathbb{E}[R_T(\tilde{\pi}_{i}(S) | \pi_{E_{-i}})] 
    + \beta \mathbb{E}[R_I(\tilde{\pi}_{i}(S) | \pi_{E_{-i}})] \\
    &> \alpha  \mathbb{E}[R_T(\pi_{E_i}(s) | \pi_{E_{-i}})] + 
\beta \mathbb{E}[R_I(\pi_{E_i}(s) | \pi_{E_{-i}})].   
\end{align*}
But by definition, for all $\pi_{E_i}(s)$,
\begin{equation*}
    \mathbb{E}[R_T(\pi_{E_i}(s) | \pi_{E_{-i}})] \geq \mathbb{E}[R_T(\tilde{\pi}_{i}(S) | \pi_{E_{-i}})]
\end{equation*}
and 
\begin{equation*}
    \mathbb{E}[R_I(\pi_{E_i}(s) | \pi_{E_{-i}})] \geq \mathbb{E}[R_I(\tilde{\pi}_{i}(S) | \pi_{E_{-i}})],
\end{equation*}
which is a contradiction.
\end{proof}

\begin{figure*}[tbh]
\centering
    \begin{subfigure}[]{0.4\linewidth}
        \centering
        \includegraphics[width=\linewidth]{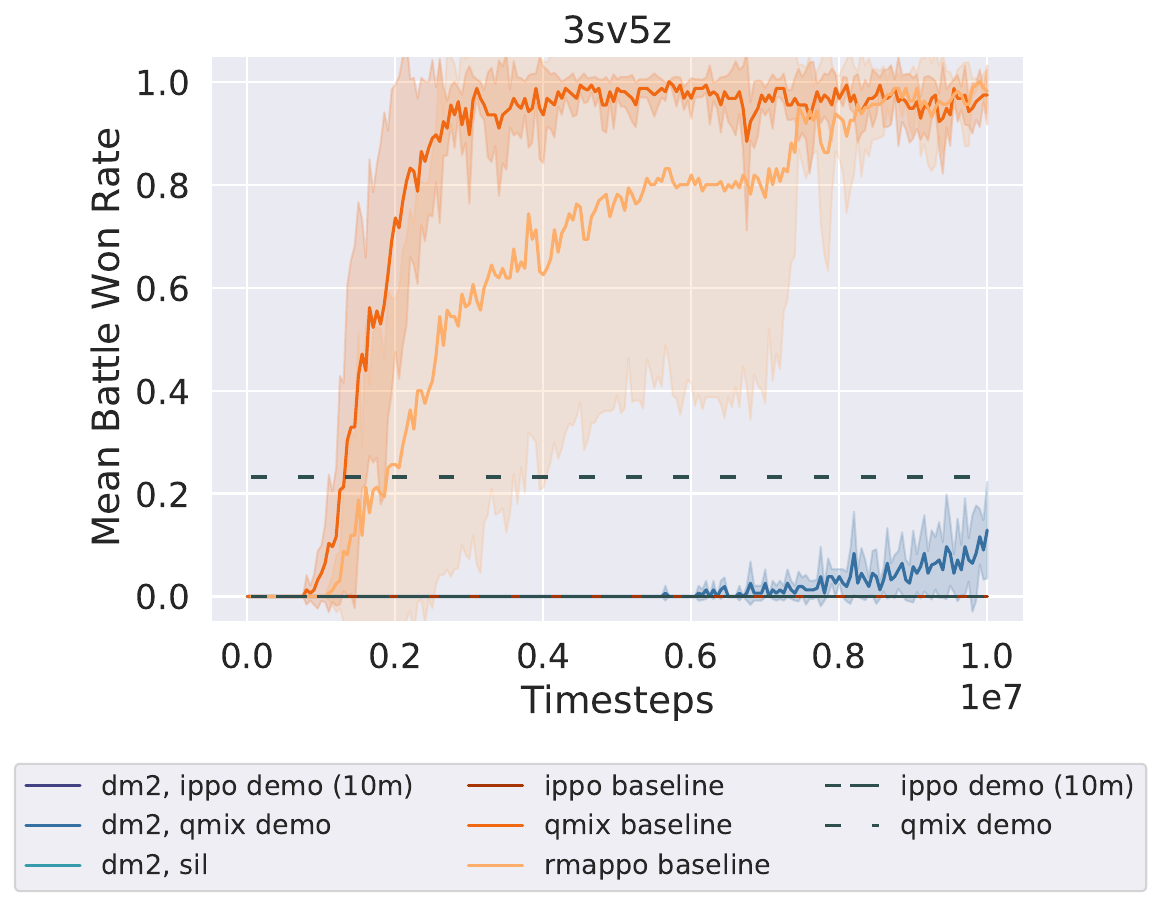}
    \end{subfigure}
    \hspace{0.8cm}
    \begin{subfigure}[]{0.4\linewidth}
        \centering
        \includegraphics[width=\linewidth]{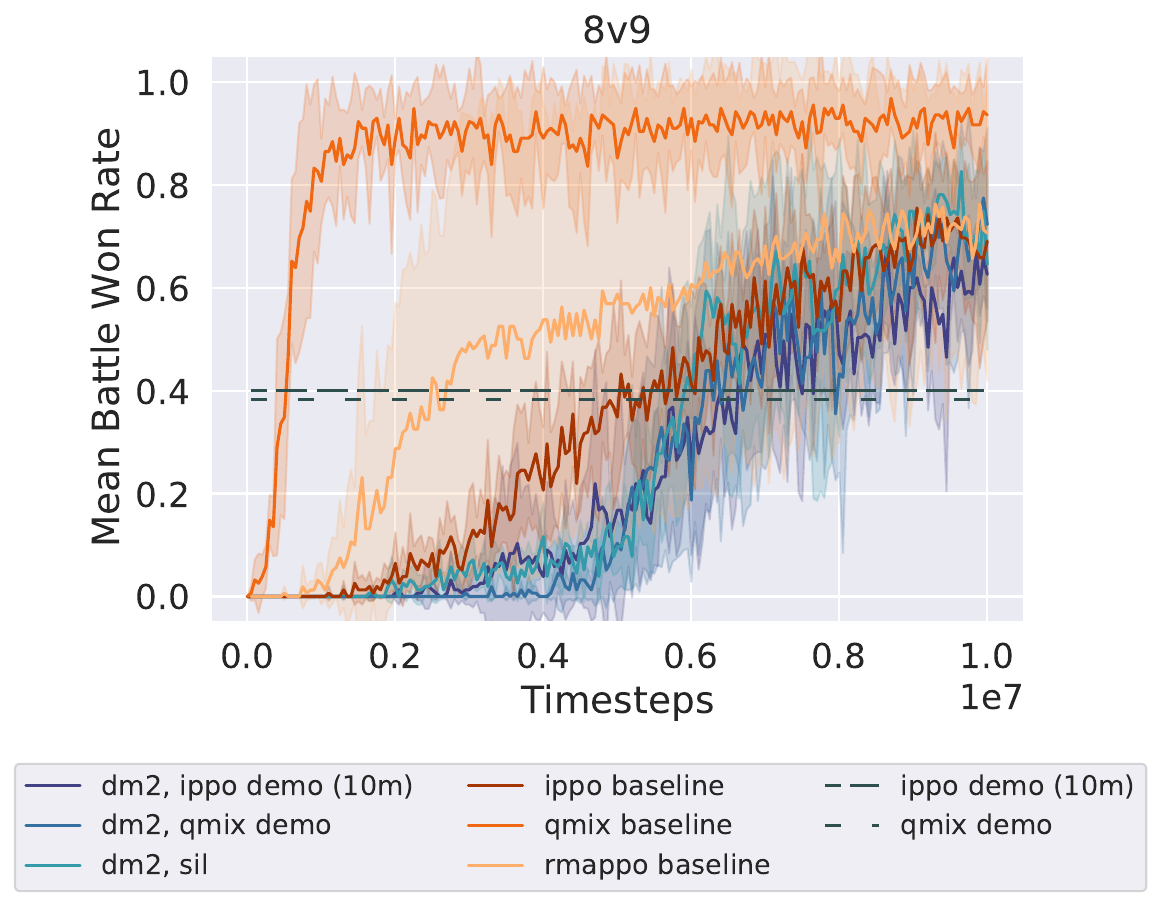}
    \end{subfigure}
\caption{DM$^2$ comparison with baselines in additional maps in the StarCraft domain. These results are an extension of the main results from the paper.}
\label{fig:more_core}
\end{figure*}

\section{Supplemental Experimental Results}
\label{app:exp_results}
This section provides further commentary on results in the main paper and presents additional, supporting experiments.
Figure \ref{fig:more_core} is an extension of our core results on two more maps in the Starcraft domain (3sv5z and 8v9), presented here due to space constraints.

\paragraph{Demonstration Styles of Ablation Study. }
The ablation study examines expert demonstrations that vary in two dimensions: co-trained versus concurrently  sampled.
For co-trained agents with demonstrations sampled non-concurrently, the demonstrations may be sampled from co-trained expert policies, but each agent's demonstrations originate from disjoint episodes.
However, for agents that were not trained together but whose demonstrations are sampled concurrently, demonstrations could be obtained from expert policies that were each trained in separate teams, but executed together in the same environment. To ensure that each expert policy is of similar quality---despite not being trained together---the joint expert policies are trained with different seeds of the same algorithm.

\paragraph{Behavioral Cloning Pretraining. } 

Distribution matching is not the only method to use demonstrations.
A more naive approach to utilize demonstrations is to use behavioral cloning (BC) \citep{bain1995framework} (a form of supervised learning) on the dataset of state-action pairs $\mathcal{D} = \langle (s_0, a_0), (s_1, a_1),  \ldots, (s_N, a_N) \rangle$. 
BC is accomplished by learning to predict the maximum likelihood action according to the dataset on the states present in the dataset.
In practice, this prediction is learned by minimizing the negative log likelihood of the expert action on these states.
This experiment pre-trains the agent policy with BC before the agents interact with the environment, after which point they learn using IPPO.

BC typically suffers from a distribution mismatch problem (also known as the covariate shift problem), where the agent’s state visitation when interacting with the environment differs from the expert data distribution, leading to poor imitation even in single-agent settings.
Behavioral cloning also requires a dataset of a size that increases quadratically with the horizon of the problem to learn successful policies \citep{ross2011reduction}.
These issues are likely to be exacerbated when dealing with multiple agents. The agents might minimize the supervised learning loss, but it is unlikely that the agents would learn to coordinate effectively.
A second issue with using the supervised learning loss to pre-coordinate agent policies is that such coordination is unlikely to last once training with IPPO proceeds.

The result for this alternative usage of the expert demonstrations is presented in Figure \ref{fig:bc_pretraining}.
As detailed above, BC does not learn to imitate the demonstrations well enough to recover their performance (indicated by the dashed lines), and as training proceeds, the benefits of BC vanish as IPPO training proceeds.

\begin{figure*}[hbtp]
\centering
    \begin{subfigure}[]{0.4\linewidth}
        \centering
        \includegraphics[width=\linewidth]{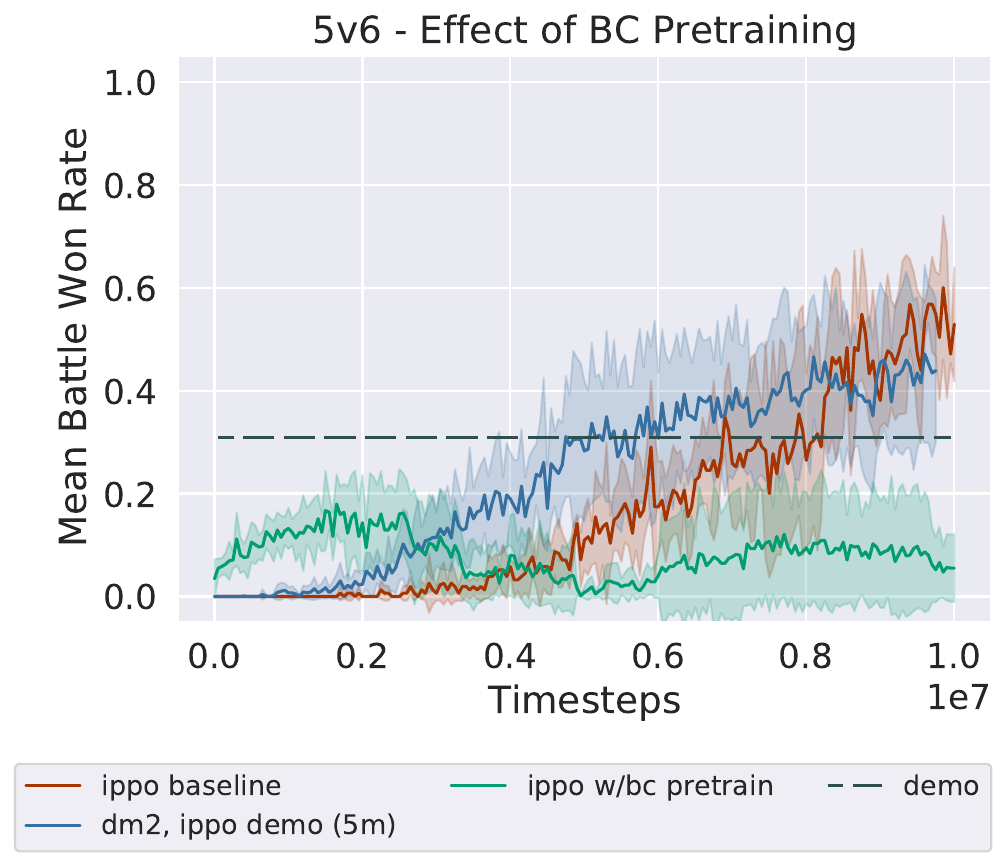}
    \end{subfigure}
    \hspace{0.8cm}
    \begin{subfigure}[]{0.4\linewidth}
        \centering
        \includegraphics[width=\linewidth]{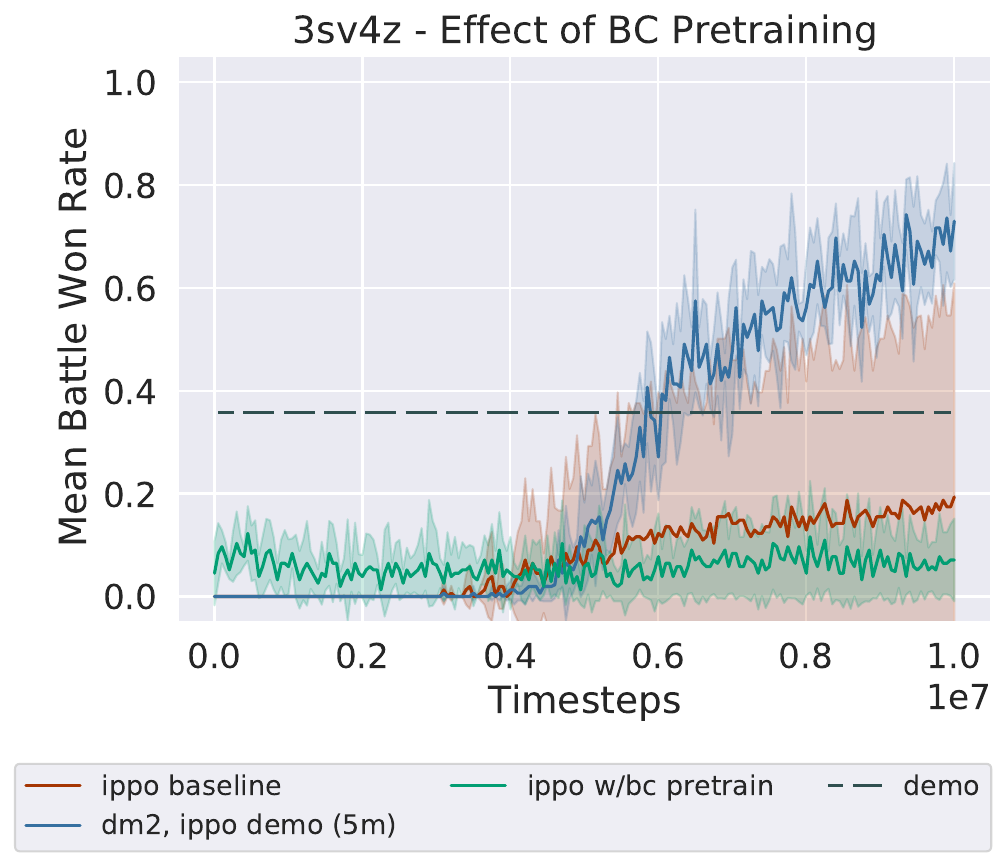}
    \end{subfigure}
\caption{IPPO trained with $r_{mix}$, with and without behavioral cloning pretraining on the expert dataset. Demonstration qualities are shown as horizontal lines. All experiments are performed with demonstrations sampled from IPPO at $5m$ timesteps of training.}
\label{fig:bc_pretraining}
\end{figure*}

\paragraph{Effect of GAIL Reward. }
DM$^2$ consists of IPPO trained with $r_{mix}$, which is a mixture of the environment and GAIL reward signal. Here, we present a brief ablation on the mixed reward (Figure \ref{fig:gail_effect}). For both IPPO demonstration qualities, IPPO trained on the GAIL reward achieves a far lower win rate than IPPO trained on the environment reward (labelled as IPPO baseline) or DM$^2$. This result provides further evidence that the primary benefit derived by DM$^2$ comes from the coordination shown in the demonstrations, rather than from individual imitation of expert behaviors. 

\begin{figure*}[t]
\centering
    \begin{subfigure}[]{0.4\linewidth}
        \centering
        \includegraphics[width=\linewidth]{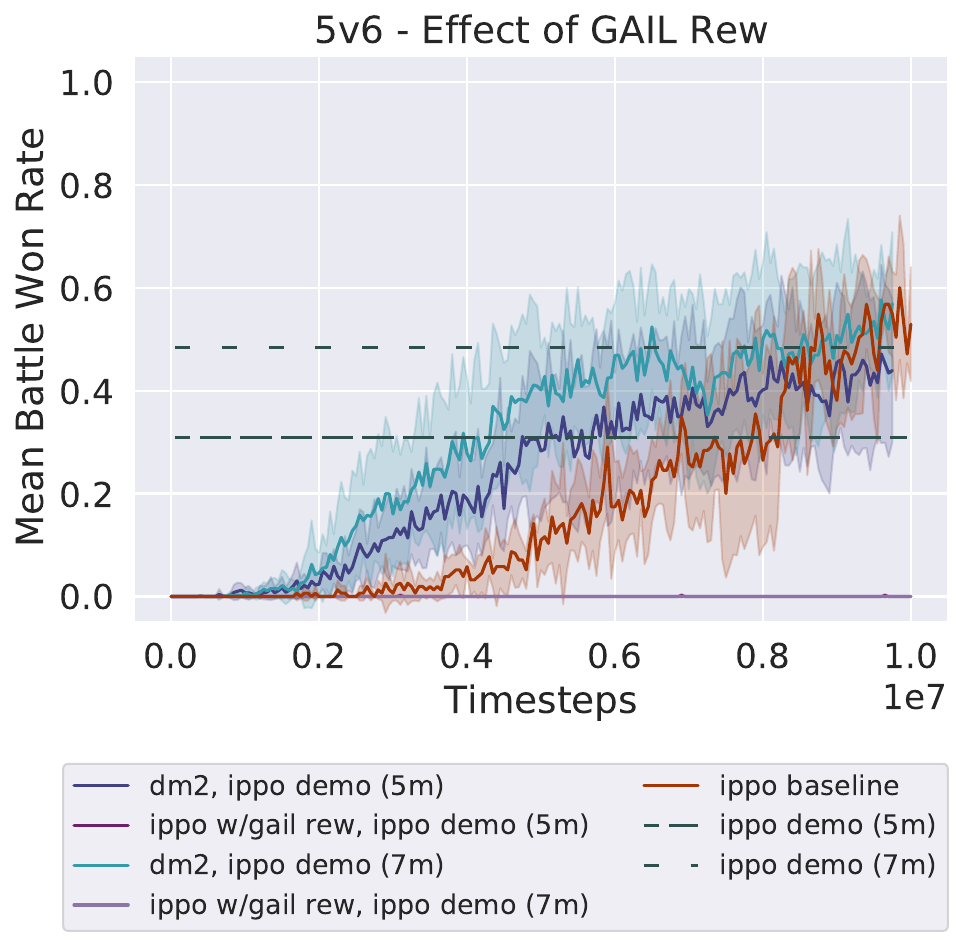}
    \end{subfigure}
    \hspace{0.8cm}
    \begin{subfigure}[]{0.4\linewidth}
        \centering
        \includegraphics[width=\linewidth]{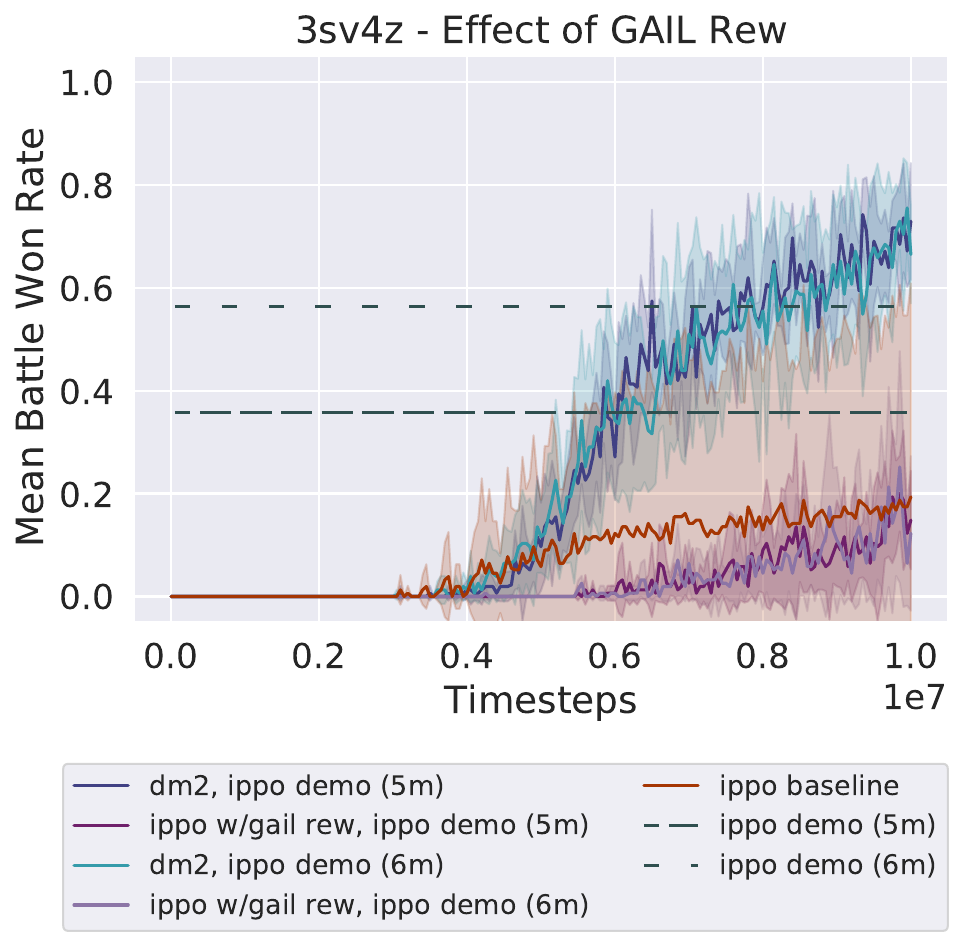}
    \end{subfigure}
\caption{IPPO trained with $r_{gail}$ and IPPO trained with $r_{env}$ only on the 5v6 and 3sv4z tasks. DM$^2$ is shown for reference. Demonstrations are sampled from IPPO policies. The win rates achieved by demonstrations are plotted as horizontal lines.}
\label{fig:gail_effect}
\end{figure*}

\paragraph{Sensitivity to Demonstration Quality.}

One finding in the experimental section of the paper is that DM$^2$ is relatively insensitive to demonstration quality beyond a certain baseline level of competence. To further support this claim, we train DM$^2$ with more demonstration qualities (where demonstrations are sampled from IPPO policies). Figure \ref{fig:demo_learning_curves_all} shows that the algorithm improves monotonically as the demonstration quality improves, but quickly saturates. This result indicates it is important to supply good demonstrations, but  not necessary to supply optimal demonstrations.

\begin{figure}[htbp]
    \centering
    \includegraphics[width=0.75\linewidth]{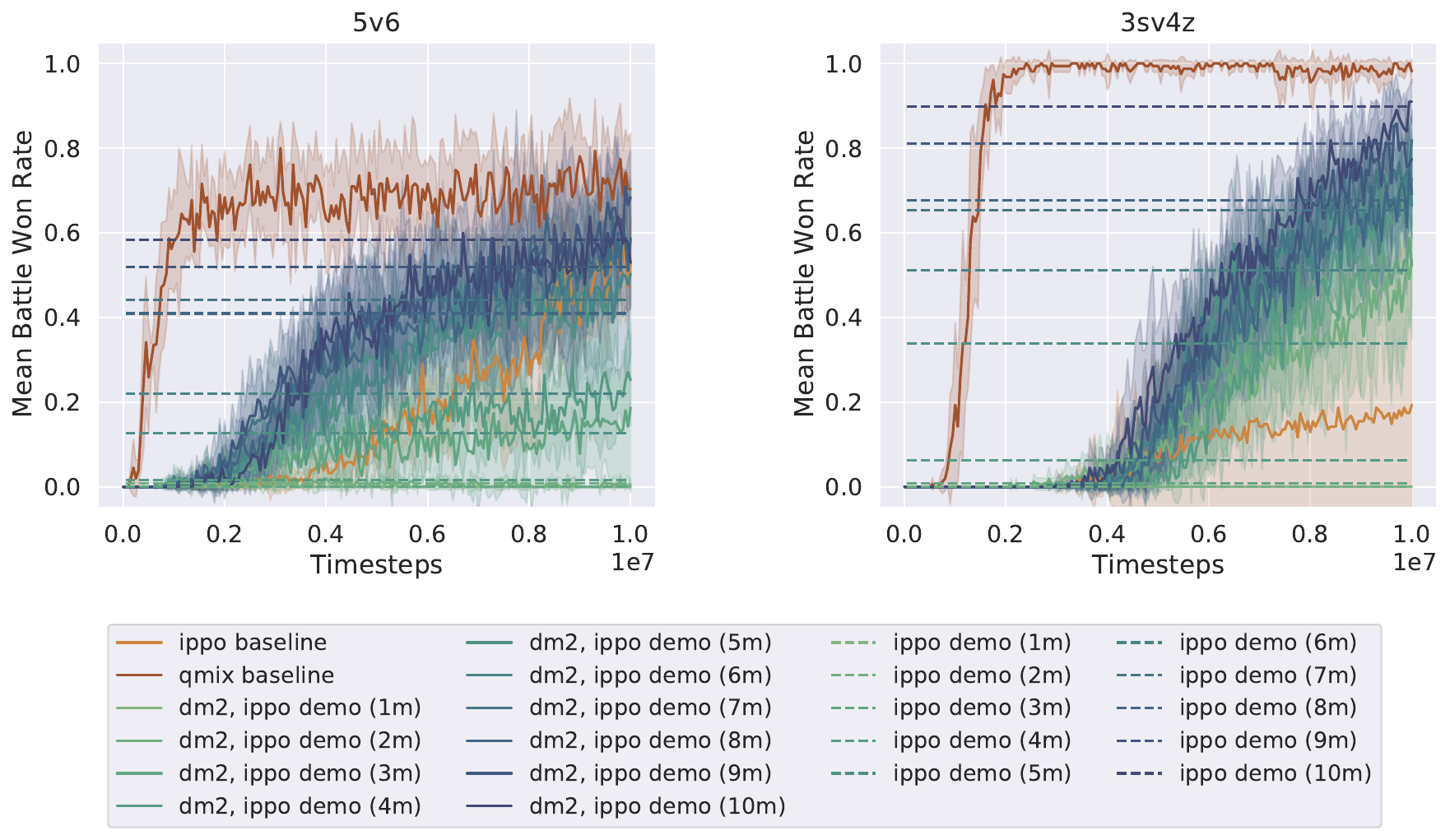}
    \caption{Learning curves of DM$^2$ (our method) trained with demonstrations sampled every million steps in the learning of the original demonstrator policy. Horizontal dotted lines indicate the demonstration qualities, colored to match corresponding learning curves. QMIX (centralized baseline) is included for reference.}
    \label{fig:demo_learning_curves_all}
\end{figure}

\section{Experimental Details} \label{app:exp_details}

\begin{wrapfigure}{r}{0.5\textwidth}
  \begin{center}
    \includegraphics[width=0.4\textwidth]{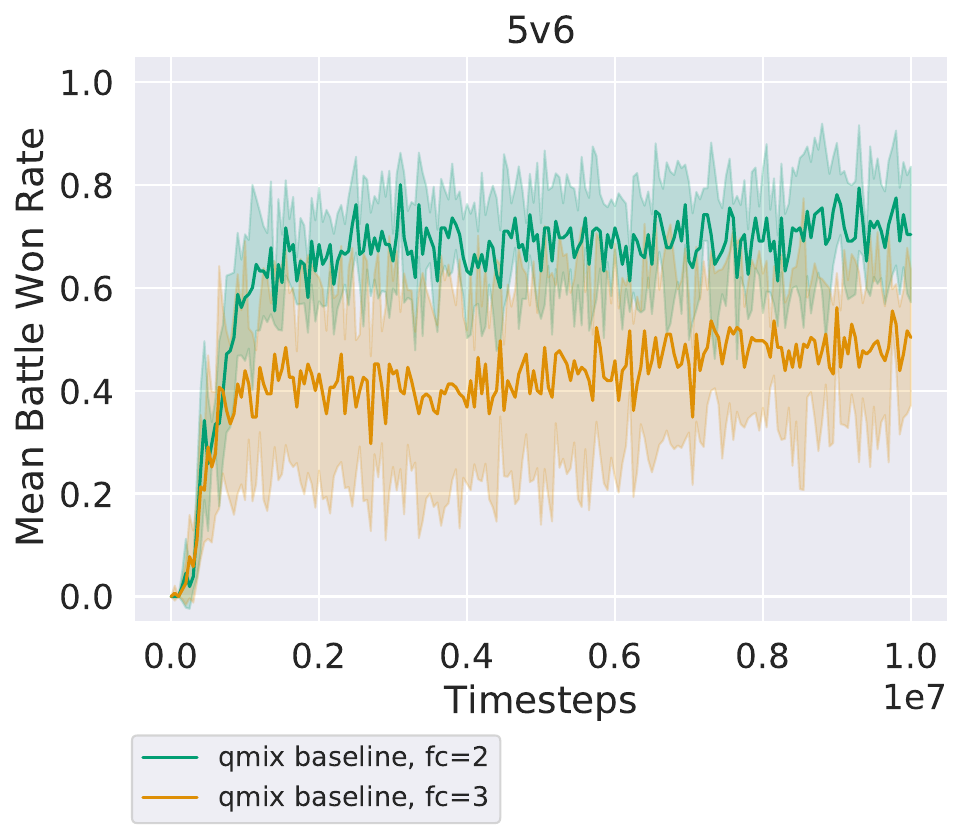}
  \end{center}
    \caption{QMIX is sensitive to the agent policy architecture. Performance on the 5v6 task suffers significantly when an extra fully connected layer is added.}
    \label{fig:qmix_policy_arch}
\end{wrapfigure}

\paragraph{Implementation Details.}

The algorithm implementations are based on the multi-agent PPO implementations provided by \citet{yu2021surprising} (MIT license) and the PyMARL code base \citep{samvelyan19smac} (Apache license). The StarCraft environment is also provided by \citet{samvelyan19smac} (MIT license). 

All decentralized MARL implementations in this paper have fully separate policy/critic networks and optimizers per agent. That is, there is no parameter sharing amongst agents. For all IPPO agents, the policy architecture is two fully connected layers, followed by an RNN (GRU) layer. Each layer has 64 neurons with ReLU activation units. The critic architecture is the same as the policy architecture. For DM$^2$ agents, the policy and critic architectures are identical to IPPO. The discriminator architecture consists of two fully connected layers with tanh activation functions. 

The centralized MARL algorithms implement agent policy networks with parameter sharing, where agents have a centralized value network. For QMIX agents, the policy architecture is the same except there is only a single fully connected layer before the RNN layer \footnote{This is the architecture used in \citet{rashid18qmix}}. We attempted running QMIX with the the IPPO agent architecture, but found that the performance of QMIX significantly suffered (Figure \ref{fig:qmix_policy_arch} on 5v6). Thus, for the QMIX experiments in the main body of the paper, the better-performing policy architecture was applied. RMAPPO agents were trained directly using the code published by \citet{yu2021surprising}.

\paragraph{Hyperparameters.}

For QMIX, the default parameters specified in \citet{rashid18qmix} are used for both tasks. For IPPO, and the IPPO component of DM$^2$, mostly default parameters (as specified in \citep{rashid18qmix, yu2021surprising}) were used. Algorithm hyperparameters that varied between tasks or were tuned are provided in Table \ref{tbl:ippo_hyperparam}. The remaining hyperparameters may be viewed with the code repository.

We found that for DM$^2$ to learn successfully from QMIX demonstrations, it was sometimes necessary to inject a small amount of random noise into the demonstration sampling process, so that the demonstrations did not constrain the exploration of the learning policies. Specifically, the demonstrations from QMIX were sampled from $\epsilon$-greedy QMIX policies, where $\epsilon$ was chosen so that the win rate did not fall more than $10\%$. QMIX $\epsilon$ values are provided in Table \ref{tbl:ippo_hyperparam}.

We conducted a hyperparameter search over the following GAIL parameters: the GAIL reward coefficient, the number of epochs that the discriminator was trained for each IPPO update, the buffer size, and the batch size. The final selected values are given in Table \ref{tbl:gail_hyperparam}.

\paragraph{Computing Architecture. }
All IPPO, QMIX, DM$^2$ experiments were performed without parallelized training; RMAPPO experiments were performed with parallelized training (as is the default in the RMAPPO codebase). The servers used in our experiments ran Ubuntu 18.04 with the following configurations:

\begin{itemize}
    \item Intel Xeon CPU E5-2698 v4; Nvidia Tesla V100-SXM2 GPU. %
    \item Intel Xeon CPU E5-2630 v4;  Nvidia Titan V GPU. %
    \item Intel Xeon Gold 6342 CPU; Nvidia A40 GPU. %
\end{itemize}

\begin{table}[hbp]
    \begin{minipage}{.5\linewidth}
    \centering
    \begin{tabular}{@{}llll@{}}
    \hline
                 & 5v6  & 3sv4z & 3sv3z \\ \hline
    epochs       & 10   & 15    & 15    \\
    buffer size  & 1024 & 1024  & 1024  \\
    gain         & 0.01 & 0.01  & 0.01  \\
    clip         & 0.05 & 0.2   & 0.2   \\
    qmix epsilon & 0    & 0     & 0.1  
    \\
    \hline 
    \\
    \end{tabular}
    \caption{IPPO Hyperparameters.}
    \label{tbl:ippo_hyperparam}
    \end{minipage}
    \begin{minipage}{.5\linewidth}
    \centering
    \begin{tabular}{llll}
    \hline
    \begin{tabular}{@{}llll@{}} 
                  & 5v6  & 3sv4z & 3sv3z \\ \hline
    gail rew coef & 0.3  & 0.05  & 0.3   \\
    discr epochs  & 120  & 120   & 120   \\
    buffer size   & 1024 & 1024  & 1024  \\
    batch size    & 64   & 64    & 64    \\
    n exp eps  & 1000 & 1000  & 1000 
    \end{tabular}
    \\ 
    \hline
    \\
    \end{tabular}
    \caption{GAIL Hyperparameters.}
    \label{tbl:gail_hyperparam}
    \end{minipage}
\end{table}